%% file: randomslopes_arxiv.tex
\documentclass{article}
\usepackage{enumitem}
\usepackage{algorithm}
\usepackage{algpseudocode}
\usepackage{multirow}
\RequirePackage{amsthm,amsmath,amsfonts,amssymb} 
\usepackage{caption,subcaption}
\usepackage[authoryear]{natbib}
\RequirePackage{relsize}
\RequirePackage[colorlinks,citecolor=blue,urlcolor=blue]{hyperref}
\usepackage[customcolors]{hf-tikz} 
\usetikzlibrary{patterns}
\usepackage{cleveref}
\usepackage{csquotes}
\usepackage{thmtools, thm-restate}

\DeclareMathOperator*{\minimize}{\text{minimize}}
\DeclareMathOperator*{\maximize}{\text{maximize}}

\input{definitions}
\usepackage{booktabs}

\usepackage{microtype}

\title{Scalable solution to crossed random effects model with random slopes}
\author{Disha Ghandwani\\Stanford University
\and Swarnadip Ghosh\\Radix Trading
\and Trevor Hastie\\Stanford University
\and Art B. Owen \\Stanford University}
\date{July 2023}
\begin{document}
\maketitle
\begin{abstract}
\input{abstract_randomslopes}
\end{abstract}
\input{main_content_randomslopes}
\section*{Acknowledgments}
We thank Brad Klingenberg and Stitch Fix for sharing
the Stitch Fix data.
This work was supported by the U.S.\ National Science Foundation
under grant IIS-1837931.

\bibliography{bigdata}
\newpage

\begin{appendix}

\section{Some proofs}
\input{proofs_appendix_random_slopes}
\section{Results under general setup}
\input{plots_appendix_randomslopes}
\end{appendix}
\end{document}

%% file: definitions.tex
\newcommand{\tran}{\mathsf{T}}
\newcommand{\Prob}{\mathbb{P}}
\newcommand{\cov}{\mathrm{Cov}}
\newcommand{\var}{\mathrm{Var}}
\newcommand{\distas}[1]{\mathrel{\overset{\text{#1}}{\sim}}}
\newcommand{\gls}{\mathrm{GLS}}
\newcommand{\Ex}{\mathbb{E}}
\newcommand{\wh}{\widehat}
\newtheorem{theorem}{Theorem}
\newtheorem{remark}{Remark}
\newcommand{\ols}{\mathrm{OLS}}
\DeclareMathOperator*{\argmin}{arg\,min}
\DeclareMathOperator*{\argmax}{arg\,max}
\DeclareMathOperator{\tr}{tr}
\newcommand{\sse}{\sigma^2_e}
\newcommand{\bsb}{\boldsymbol{b}}
\newcommand{\bsa}{\boldsymbol{a}}

\newcommand{\bbeta}{\boldsymbol{\beta}}
\newcommand{\btheta}{\boldsymbol{\theta}}
\newcommand{\balpha}{\boldsymbol{\alpha}}
\newcommand{\OmegaO}{O}

%% file: abstract_randomslopes.tex
The crossed random effects model is widely used, finding applications
in various fields such as longitudinal studies, e-commerce, and
recommender systems, among others. However, these models encounter
scalability challenges, as the computational time for standard
algorithms grows superlinearly with the number $N$ of observations in
the data set, commonly $\OmegaO(N^{3/2})$ or worse. Recent published works present scalable methods for crossed random effects in linear models and some generalized linear models, but those methods only allow for random intercepts. In this paper, we devise scalable algorithms for models that include random slopes. This addition brings substantial difficulty in estimating the random-effect covariance matrices in a scalable way. We address this issue by using a variational EM algorithm. Our proposed approach accommodates both diagonal covariance matrices and cases where no structure is assumed—a scenario common in fields such as psychology and neuroscience.
In simulations, the proposed method is substantially faster than
standard methods for large $N$. It is also more efficient than
ordinary least squares which has a problem of greatly
underestimating the sampling uncertainty in parameter estimates. We
illustrate the new method on a \emph{MovieLens} data set, as well as a large data set (five million observations) from the online retailer \emph{Stitch Fix}.

%% file: main_content_randomslopes.tex
\section{Introduction}\label{sec.introduction}
In the era of \emph{over the top} (OTT) platforms like Netflix, where
customer acquisition is booming and a vast array of titles are hosted,
the need for computationally efficient recommender systems has become
crucial. In such scenarios, it is common for different clients to rate
varying numbers of items, which gives a haphazard missingness
pattern. We deal with one such problem where $R$ clients rate $C$
items, and it is usually the case that only $N \ll (R \times C)$ of
the ratings are observed. We consider the situation where we have
access to various characteristics of each (item, client) pair as well
as  the client's rating of the item. For instance, with the online
clothing retailer we have information such as the client's age, sex,
and purchase history, as well as item features such as material and
clothing style. For the movie data, we have genre, language, and
duration for the online content.

Our motivating data set comes from the
online clothing retailer Stitch Fix. Quoting from~\cite{ghos:hast:owen:logistic2022}:\\
\begin{quotation}
Stitch Fix is an online personal styling service.
One of their business models involves sending
customers a sample of clothing items. The customer may
keep and purchase any of those items and
return the others.
They have provided us
with some of their client ratings data.
That data was anonymized, void of personally identifying
information, and as a sample it does not reflect their total
numbers of clients or items at the time they provided it.
It is also from 2015. While it does not describe their
current business, it is a valuable data set for illustrative purposes.
\end{quotation}
We present results using these data in Section~\ref{sec:appl-cross-rand}, and we use the problem
setup throughout this paper to explain our choices.
The goal  is to model clients' ratings of items.
Let $y_{ij}$ represent the rating provided by client $i$ for item $j$,
and $\mathbf{x}_{ij}$ (typically a vector)  the covariate information for the
specific client-item pair $(i,j)$.
In such settings, ordinary least squares is not advisable as the
ratings are not independent. Instead, we expect some positive
correlation between ratings by different clients on the same item,
because some items are simply more popular than others in ways that
cannot always be attributed to $\mathbf{x}_{ij}$. We can similarly expect some
positive correlation between ratings by the same client on different items
because clients vary in their preferences.
To account for such covariance structure, crossed random-effects models are employed. These models are designed to capture the dependencies present in the ratings provided by the same client or for the same item.
One notable drawback of such  models is their limited scalability. As the data set grows in size, the computational burden associated with fitting crossed random effects models becomes challenging and inefficient.
~\cite{ghos:hast:owen:2021} considered the crossed random-effects model with random intercept terms for clients and items, given by
\begin{equation}
\label{eq:intercept_model}
y_{ij} = \beta_0+ a_i +b_j+\textbf{x}_{ij}^{\tran}\mathbf{\bbeta} + \varepsilon_{ij}, \hspace{0.5cm} \forall i \in \{1,\cdots, R\}, \hspace{0.3cm} j\in \{1,\cdots, C\}.
\end{equation}
Here, the random effects are $a_i \in \mathbb{R}$, $b_j \in \mathbb{R}$ and an error $ \varepsilon_{ij} \in \mathbb{R}$. The fixed effects are the regression parameters $\mathbf{\bbeta} \in \mathbb{R}^{p}$ and the intercept $\beta_0 \in \mathbb{R}$. They assume that $a_i \sim \mathcal{N}(0, \sigma^2_a)$ , $b_j \sim \mathcal{N}(0, \sigma^2_b), $
and $\varepsilon_{ij}\sim \mathcal{N}(0,\sigma^2_e)$ are all independent. Under this model,
$$\cov(y_{ij}, y_{ij'}) = \sigma^2_a \mbox{ and } \cov(y_{ij}, y_{i'j}) = \sigma^2_b \mbox{ for } i \neq i' , j \neq j'.$$

Standard approaches for fitting such models (including ``lmer" function of ``lme4" library in R) come with a computational
burden proportionate to $(R+C)^3$, where $R$ and $C$ are the number of
clients and items respectively. However, due to the relationship $RC
\ge N$, this computational cost can be expressed as
$\OmegaO(N^{3/2})$. To address this limitation, \cite{crelin} introduced a moment-based method for estimating variance parameters in the intercept model ($\ref{eq:intercept_model}$) with a computational cost that is at most linear. More recently, \cite{ghos:hast:owen:2021} developed a backfitting algorithm inspired by the work of \cite{buja:hast:tibs:1989} that achieves linear computational cost in terms of $N$.
In contrast, the plain Gibbs sampler is not scalable, as it exhibits a
computational cost of $\OmegaO(N^{3/2})$, as shown by \cite{crevarcomp}. They also show that many other Bayesian approaches do not work well in crossed random effects settings. Recent work has shown that collapsed Gibbs sampling can be made to scale well in crossed random effects settings. See \cite{papa:robe:zane:2020} and \cite{ghos:zhon:2021:tr}.
 \cite{ghosh:thesis} presents an insightful connection between the convergence
of backfitting and collapsed Gibbs. 
These works do not include the random slopes that we focus on here. 

During the process of model fitting, \cite{ghos:hast:owen:2021} also
produce (posterior or BLUP) estimates for $(a_i)_{i=1}^R$ and $(b_j)_{j=1}^C$, which can be used for making predictions on test data. This implies that the crossed random effects model can serve as a scalable recommender system. To assess its prediction accuracy, we measured the performance of the crossed random effects model with a random intercept on Stitch Fix data using a random train-test split. In comparison to ordinary least squares, which yielded an $R^2$ of 4.61$\%$, the crossed random effects model (\ref{eq:intercept_model}) achieved an impressive $R^2$ of 18.42 $\%$. This significant improvement in $R^2$
shows the efficiency gain from
crossed random effects models. We also demonstrate the gains achieved
on the \emph{MovieLens} data.

These developments motivated us to develop an
efficient and scalable implementation for crossed-random effects with
both random intercepts and slopes.
\begin{equation}
\label{eq:random_slopes_uncompact}
y_{ij} = \beta_0+ a_i + b_j + \mathbf{x}_{ij}^{\tran}(\bbeta+\tilde{\bsa}_i+\tilde{\bsb}_j) + \varepsilon_{ij}, \hspace{0.3cm} \forall i \in \{1,\cdots, R\}, \hspace{0.3cm} j\in \{1,\cdots,C\}, \end{equation}
where we have augmented (\ref{eq:intercept_model}) with  random slopes
$\tilde{\bsa}_i \in \mathbb{R}^{p}$ and $\tilde{\bsb}_j \in
\mathbb{R}^{p}$. We assume that $\tilde{\bsa}_i \sim \mathcal{N}_p(0,
\tilde{\Sigma}_a)$ and $\tilde{\bsb}_j \sim \mathcal{N}_p(0,
\tilde{\Sigma}_b)$. We further assume that $(a_i,
  \tilde{\bsa}_i)$ follows a multivariate gaussian distribution and is
  independent of $(b_j, \tilde{\bsb}_j)$. The crossed random effects
model with random slopes helps us capture the variability of the
effect size of covariates among different clients and
items. Figure~\ref{fig:random_slope_intuition} provides a visual
representation of the model with random intercept and random
slopes.  The first plot
represents a random effects model with a random intercept
for clients and a fixed slope for the quantitative feature ``match'';
the second plot has a random intercept as well as a random
slope for ``match".
\begin{figure}[ht]
\centering
\includegraphics[width = \linewidth]{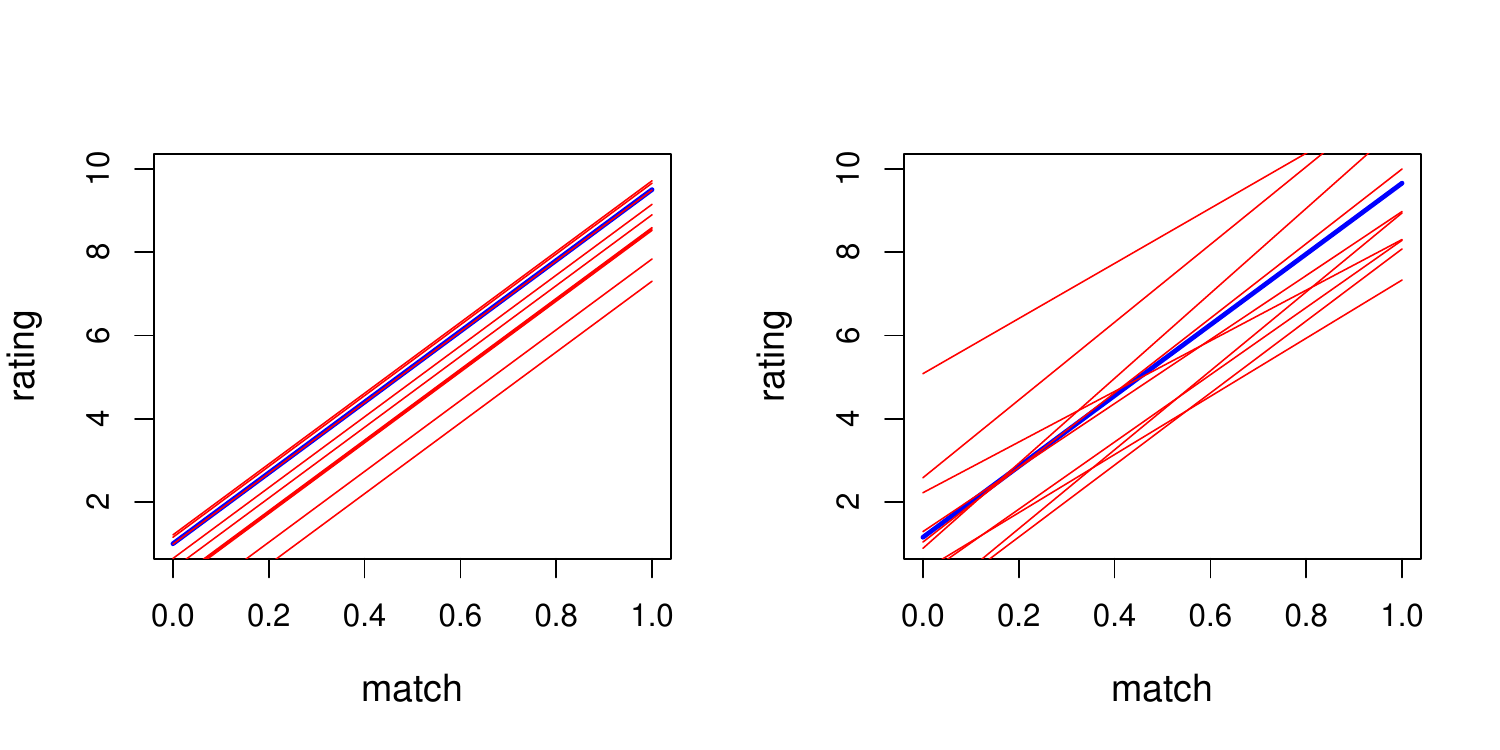}
\caption{In the left plot, we illustrate a model with random
  intercepts and a fixed slope for the quantitative variable
  ``match''. In the right plot, we have random intercepts as well as
  random slopes for ``match''. The blue line represents the fixed
  intercept and slope in both plots.}
\label{fig:random_slope_intuition}
\end{figure}
We can simpify the notation in ($\ref{eq:random_slopes_uncompact}$) by
including 1 as the first variable in $\textbf{x}_{ij}$ to accommodate the intercept:
\begin{equation}
\label{eq:random_slopes}
y_{ij} = \textbf{x}_{ij}^{\tran}(\bbeta + \bsa_i +\bsb_j) + \varepsilon_{ij}, \hspace{0.5cm} \forall i \in \{1,\cdots, R\}, \hspace{0.3cm} j\in \{1,\cdots,C\}.
\end{equation}
Now $\textbf{x}_{ij} \in
\mathbb{R}^{p+1}$, and we assume that $\bsa_i \sim \mathcal{N}_{p+1}(0, \Sigma_a)$ , $\bsb_j \sim \mathcal{N}_{p+1}(0, \Sigma_b), $
and $\varepsilon_{ij} \distas{i.i.d} \mathcal{N}(0,\sigma^2_e)$ are all
independent. Hence the random effects include both
  intercepts and slopes.\\
\begin{remark}
\label{remark:randomslopes}
For notational simplicity, we have included random slopes for all covariates when describing our method.   In practice, several variations are attractive and easily accommodated\footnote{The R and Python code provided at \url{https://github.com/dishaghandwani/randomslopes} handles these variations.}: 
\begin{itemize}
\item We may want random slopes for just a subset of the covariates, depending on their nature. 
\item We may want different subsets of covariates having  random
  slopes for clients than for items. 
\item Some covariates are specific to the client (e.g., age or gender) and do not change with different items. Therefore, client random slopes do not make sense and should not be included. However, item random slopes can still be used. The reverse applies to item-specific covariates
\end{itemize}
\end{remark}
The goal is to estimate the fixed effect $\bbeta$ and covariance parameters $\Sigma_a$, $\Sigma_b$, and $\sigma^2_e$. When $\bsa_i, \bsb_j$
and $\varepsilon_{ij}$ are all Gaussian in (\ref{eq:random_slopes}), the likelihood for ($\bbeta, \Sigma_a, \Sigma_b, \sigma^2_e$) is
\begin{equation}
\label{eq:likelihood}
\frac{1}{\sqrt{(2\pi)^N|\mathcal{V}|}} \exp\left(-\frac{1}{2}(y - X \bbeta )^{\tran} \mathcal{V}^{-1}(y-X \bbeta)\right).
\end{equation}
Here, the vector $Y$ denotes all of the $y_{ij}$ placed into a vector
in $\mathbb{R}^N$ and $X$ denotes $\mathbf{x}_{ij}$ stacked compatibly into a
matrix in $ \mathbb{R}^{N \times (p+1)}$.
$\mathcal{V} \in \mathbb{R}^{N \times
  N}$ represents the covariance matrix of
the $N$-vector $Y$, which is a somewhat complex function of ($\bbeta, \Sigma_a, \Sigma_b, \sigma^2_e$):
\begin{equation}
\label{eq:V}
\mathcal{V} = \left( X \Sigma_a X^{\tran} \right) \bullet \left(Z_a Z_a^{\tran}\right) + \left(X \Sigma_b X^{\tran}\right) \bullet \left(Z_b Z_b^{\tran}\right) + \sigma^2_e I_N, \end{equation}
where $Z_a \in \lbrace 0,1 \rbrace^{N \times R}$ denotes the indicator
(one-hot encoding)
matrix for clients, $Z_b \in \lbrace 0,1 \rbrace^{N \times C}$ denotes
the indicator matrix for items and $A \bullet B$ denotes elementwise (Hadamard)
product of matrix $A$ and $B$. 
Let data point $t_1$ refer to the
client-item pair $(i_1,j_1)$ and data point $t_2$ refer to the pair $(i_2,j_2)$, then
\begin{equation}\label{eq:1}
\mathcal{V}_{t_1,t_2} =
\begin{cases}
0, & \text{if } i_1 \neq i_2 \text{ and } j_1 \neq j_2, \\
x_{i_1 j}^{\tran} \Sigma_b x_{i_2j} , & \text{if } i_1 \neq i_2 \text{ and } j_1 = j_2 = j,\\
x_{ij_1}^{\tran} \Sigma_a x_{ij_2} , & \text{if } i_1 = i_2 = i \text{ and } j_1 \neq j_2,\\
\mathbf{x}_{ij}^{\tran} \Sigma_a \mathbf{x}_{ij}^{\tran} + \mathbf{x}_{ij}^{\tran} \Sigma_b \mathbf{x}_{ij} + \sigma^2_e , & \text{if } (i_1, j_1) = (i_2, j_2) = (i,j).\\
\end{cases}\end{equation}
The bottom term in (\ref{eq:1}) represents a diagonal term in
$\mathcal{V}$, i.e. the variance of $y_{ij}$.
The maximum likelihood estimate (MLE) of
$\bbeta$ is also the generalized least squares (GLS) estimate.
Hence if $\mathcal{V}$ were known (i.e. we have estimates for
its ingredients), we could directly estimate $\bbeta$ by GLS.
However, since $\mathcal{V}$  is generally dense and not usefully structured, obtaining $\hat\bbeta_{\gls}$ naively by GLS would  be computationally
prohibitive ($\mathcal{O}(N^3)$).
\begin{itemize}
\item
  Our first approach tackles this GLS problem using method-of-moment
  estimates for the covariance parameters. In Section \ref{sec.backfitting}, we describe an efficient
  backfitting procedure for estimating $\bbeta$ by GLS if the the
  covariance parameters are known. We also provide consistent
  estimates of the covariance parameters using method-of-moments, as
  described in Section~\ref{sec.covariance}. Both of the tasks can be
  accomplished in $\mathcal{O}(N)$ time.
\end{itemize}
The drawback of this first approach is that when  $p$ is large and the
covariance matrices are non-diagonal, the method of moments may not
produce accurate estimates for the covariance parameters.
Direct
maximum likelihood estimation of all the parameters using
(\ref{eq:likelihood}) seems out of the question.
\begin{itemize}
\item In our second approach, we maximize an approximation to the likelihood~(\ref{eq:likelihood})
using \emph{variational EM} to estimate the fixed effects and covariance parameters simultaneously. It helps us in estimating covariance matrices precisely even when they are unstructured or $p$ is sufficiently large. Another advantage of this approach is its computational efficiency, with each step at most $\mathcal{O}(N)$.
\end{itemize}
The variational EM algorithm is closely related to the backfitting algorithm. Variational EM has been previously explored for fast and accurate estimation of non-nested binomial hierarchical models by~\cite{max:2022}.
We speed up the convergence rate for both of the approaches ---
backfitting and variational EM --- using a technique referred to as ``clubbing'' discussed in Section \ref{sec.backfitting}.
\subsection{Paper Outline}
\begin{itemize}
\item In Section \ref{sec.backfitting}, we discuss the backfitting approach with and without clubbing.
\item In Section \ref{sec.covariance}, we discuss the estimation of covariance parameters using the method of moments in two scenarios: when the covariance matrices, $\Sigma_a$ and $\Sigma_b$ are diagonal, and when they are non-diagonal.
\item In Section \ref{sec.variational}, we discuss the application of
  variational EM in estimating fixed as well as covariance parameters
  using an approximate maximum likelihood approach.
\item In Section~\ref{sec:results}, we show the performance of these methods
  using both simulated and real data.
\end{itemize}
\section{Estimation of \texorpdfstring{$\bbeta$}{beta} by GLS when the covariance parameters are known}\label{sec.backfitting}
If the covariance parameters are known, the aim is to obtain $ \hat\bbeta_{\gls}$ in $\mathcal{O}(N)$ time. We know that the GLS estimate of $\bbeta$ is
\begin{equation}
\label{eq:beta_gls}
\hat\bbeta_{\gls} = \argmin_{\bbeta} (Y-X\bbeta)^\tran\mathcal{V}^{-1}(Y -X\bbeta)=(X^{\tran}\mathcal{V}^{-1}X)^{-1}X^{\tran}\mathcal{V}^{-1}Y .
\end{equation}
But, we couldn't use the above formula directly to obtain $
\hat\bbeta_{\gls}$ as computing $\mathcal{V}^{-1} $ would take at least
$\mathcal{O}(N^{3/2})$ time in the above setting (and more generally
$\mathcal{O}(N^{3})$). The theorem below helps us in setting up the
motivation for the $\mathcal{O}(N)$) backfitting algorithm.
\begin{restatable}{theorem}{robin}
\label{thm:robinson}
Consider the solutions to the following penalized least-squares problem
\begin{equation}
\label{eq:loss}
\min_{\bbeta,\bsa_i,\bsb_j} \sum_{i=1}^R \sum_{j=1}^{C}\frac{z_{ij}(y_{ij} - \mathbf{x}_{ij}^{\tran}(\bbeta + \bsa_i +\bsb_j))^2}{\sigma^2_e} + \sum_{i=1}^R \bsa_i^{\tran} \Sigma^{-1}_a \bsa_i + \sum_{j=1}^C \bsb_j^{\tran} \Sigma_b^{-1} \bsb_j .
\end{equation}
Then $\hat\bbeta = \hat\bbeta_{\gls}$ and the $\{\hat\bsa_i\}_{i=1}^{R}$ and $\{\hat\bsb_j\}_{j=1}^C$ are the best linear unbiased prediction (BLUP) estimates of the random effects if the covariance parameters $\Sigma_b$, $\Sigma_b$ and $\sigma^2_e$ are fixed and known.
\end{restatable}
Here we have introduced a binary \emph{masking} matrix $Z \in R \times C$ for
ease of notation:  $z_{ij} = 1$ if the rating $y_{ij}$ for the $(i,j)$
client-item pair is observed, and $z_{ij} = 0$ otherwise. 
\begin{remark}
In practice, we may also observe multiple observations for client-item
pairs, which entails additional notation. Our algorithms are able to
accomodate this seamlessly, as the code mostly involves grouping the observations by items or clients. In order to keep the notation simple, we restrict our presentation to the binary annototation of $z_{ij}$. Another important thing to note here is that, although the first term in \eqref{eq:loss} appears as a summation of $RC$ values, it is just for illustrative purpose. The summation is actually performed only on observed ratings. 
\end{remark}
The theorem was originally given by~\cite{robinson91:_that_blup}; we
provide an alternative proof in Appendix~\ref{Alternative_proof}. If the covariance
parameters are known or somehow estimated, the fixed-effect parameter
can be estimated by backfitting algorithms (block-coordinate descent) which minimize ($\ref{eq:loss}$). We estimate $\bbeta$ using backfitting with and without clubbing similar to~\cite{ghos:hast:owen:2021} and~\cite{ghos:hast:owen:logistic2022}.
\subsection{Vanilla Backfitting} The objective function
$(\ref{eq:loss})$ treats the three groups $\{\bsa_i\}_{i=1}^R$, $\{\bsb_j\}_{j=1}^C$, and
$\bbeta$ as parameters to be estimated by least squares,  with  $\{\bsa_i\}_{i=1}^R$ and $\{\bsb_j\}_{j=1}^C$
regularized via  generalized ridge penalties.   The idea of
backfitting is to cycle through the estimation of the parameters in each group, keeping the parameters fixed in the remaining groups, i.e., a block coordinate descent. Minimizing (\ref{eq:loss}) is equivalent to minimizing
\begin{equation}\sum_{i=1}^R \sum_{j=1}^{C}z_{ij}(y_{ij} - \mathbf{x}_{ij}^{\tran}(\bbeta + \bsa_i +\bsb_j))^2 + \sum_{i=1}^R \bsa_i^{\tran} \Lambda \bsa_i + \sum_{j=1}^C \bsb_j^{\tran} \Gamma \bsb_j , \end{equation}
where $\Lambda = \sigma^2_e \Sigma^{-1}_a$ and $\Gamma = \sigma^2_e \Sigma^{-1}_b$ are the precision matrices corresponding to the two random effects.
Thus, the fit for $\bsa_i$ keeping $\{\bsb_j\}_{j=1}^{C}$ and $\bbeta$ fixed is given by
\begin{equation}
\label{eq:(a_i)}
\hat {\bsa}_i = \bigg(\sum_{j=1}^{C} z_{ij}\mathbf{x}_{ij} \mathbf{x}^{\tran}_{ij}+ \Lambda\bigg)^{-1}\bigg(\sum_{j=1}^{C}z_{ij}\mathbf{x}_{ij} \Big(y_{ij} - \mathbf{x}_{ij} ^{\tran}(\bbeta +\bsb_j)\Big)\bigg) .
\end{equation}
Similarly, fit for $\bsb_j$ keeping $\{\bsa_i\}_{i=1}^{R}$ and $\bbeta$ fixed is given by
\begin{equation}
\label{eq:(b_j)}
\hat {\bsb}_j = \bigg(\sum_{i=1}^{R} z_{ij} \textbf{x}_{ij}\textbf{x}_{ij}^\tran+ \Gamma\bigg)^{-1}\bigg(\sum_{i=1}^{R} z_{ij}\mathbf{x}_{ij} \Big(y_{ij} - \mathbf{x}_{ij} ^{\tran}(\bbeta +\bsa_i)\Big)\bigg) .
\end{equation}
The fit for $\bbeta$ given $\bsa_i$ and $\bsb_j$ is given by
\begin{equation}
\label{eq:beta}
\hat\bbeta = \bigg(\sum_{j=1}^C \sum_{i=1}^{R}z_{ij} \mathbf{x}_{ij} \mathbf{x}^{\tran}_{ij}\bigg)^{-1}\bigg(\sum_{j=1}^C \sum_{i=1}^{R} \Big(y_{ij} - \mathbf{x}_{ij} ^{\tran}(\bsa_i+\bsb_j)\Big)z_{ij}\mathbf{x}_{ij} \bigg) .
\end{equation}
We cycle through these equations until the fits converge.
\begin{remark}
$\left.\right.$
\begin{itemize}
\item For each group of parameters, we form the partial residuals holding the others fixed. Then, we optimize with respect to that group. For the random effect parameters, the problem decouples, as each step in $(\ref{eq:(a_i)})$ and $(\ref{eq:(b_j)})$ indicate.
\item Convergence is measured in terms of the fits of the terms, $\mathbf{x}_{ij} ^{\tran} \bbeta, \{\mathbf{x}_{ij} ^{\tran} \bsa_i \}$ and $ \{\mathbf{x}_{ij} ^{\tran} \bsb_j\}$.
\item Each step decreases the objective, and the algorithm converges to the solution to (\ref{eq:loss}).
\item Notice that in $(\ref{eq:(a_i)})$, for fixed $i$, we only visit terms
  $j\in \{1,2,\ldots,C\}$
  for which $z_{ij}=1$. Hence, we visit in total $N$ terms for all $i$.
  Likewise, this applies to $(\ref{eq:(b_j)})$ and $(\ref{eq:beta})$. Hence the
  computational cost of obtaining all the above fits is $\mathcal{O}(N)$. 
\end{itemize}
\end{remark}
\subsection{Backfitting with clubbing}
One of the shortcomings of the vanilla backfitting algorithm is that
it can take a long time to converge. One of the reasons is that the solution has built-in constraints that are not visible to the individual steps. The theorem below provides two such constraints.
\begin{theorem}
The solutions $\{\bsa_i\}_{i=1}^R$ and $\{\bsb_j\}_{j=1}^C$ of $(\ref{eq:loss})$ satisfy $\sum_{i=1}^R \hat{\bsa}_i = 0 $ and $ \, \sum_{j=1}^C \hat{\bsb}_j = 0 $.
\end{theorem}
\begin{proof} It is enough to prove $\sum_{i=1}^R \hat{\bsa}_i = 0 $, the proof for $ \, \sum_{j=1}^C \hat{\bsb}_j = 0 $ is similar. Suppose $\sum_{i=1}^R \hat{\bsa}_i \ne 0 $, define $ \bar{\bsa} = \sum_{i=1}^R \hat{\bsa}_i / R$ and $\tilde{\bsa}_i = \hat{\bsa}_i - \bar{\bsa} \, \, \forall i$. Then it is easy to show that $$\sum_{i=1}^R \tilde{\bsa}_i^{\tran} \Sigma^{-1}_a \tilde{\bsa}_i = \sum_{i=1}^R (\hat{\bsa}_i - \bar{\bsa})^{\tran} \Sigma^{-1}_a (\hat{\bsa}_i - \bar{\bsa}) < \sum_{i=1}^R \hat{\bsa}_i^{\tran} \Sigma^{-1}_a \hat{\bsa}_i,$$ unless $\bar{\bsa} = 0$.
If we choose, $\tilde \bbeta = \hat\bbeta + \bar{\bsa}$, i.e., $\tilde \bbeta + \tilde{\bsa}_i = \hat\bbeta + \hat{\bsa}_i$, then \begin{equation} \sum_{i=1}^R \sum_{j=1}^{C}\frac{z_{ij}(y_{ij} - \mathbf{x}_{ij} ^{\tran}(\tilde \bbeta + \tilde{\bsa}_i +\bsb_j))^2}{\sigma^2_e} = \sum_{i=1}^R \sum_{j=1}^{C}\frac{z_{ij}(y_{ij} - \mathbf{x}_{ij} ^{\tran}(\hat\bbeta + \hat{\bsa}_i +\bsb_j))^2}{\sigma^2_e}. \end{equation}
Thus,
\begin{align*}
\min_{\bbeta,\bsb_j} \sum_{i=1}^R \sum_{j=1}^{C} & \frac{z_{ij}(y_{ij} - \mathbf{x}_{ij} ^{\tran}(\bbeta + \tilde{\bsa}_i +\bsb_j))^2}{\sigma^2_e} + \sum_{i=1}^R \tilde{\bsa}_i^{\tran} \Sigma^{-1}_a \tilde{\bsa}_i \\
& < \min_{\bbeta,\bsb_j} \sum_{i=1}^R \sum_{j=1}^{C}\frac{z_{ij}(y_{ij} - \mathbf{x}_{ij} ^{\tran}(\bbeta + \hat{\bsa}_i +\bsb_j))^2}{\sigma^2_e} + \sum_{i=1}^R \hat{\bsa}_i^{\tran} \Sigma^{-1}_a \hat{\bsa}_i,
\end{align*}
unless $\bar {\bsa} = 0$ which leads to contradiction.
\end{proof}
The theorem above shows that the constraint $\sum_{i=1}^R \hat{\bsa}_i =
0 $ is automatically imposed if we minimize (\ref{eq:loss})
simultaneously in terms of $\{\bsa_i\}_{i=1}^{R}$ and $\bbeta$. Similarly,
the constraint $\sum_{j=1}^C \hat{\bsb}_j = 0 $ is automatically imposed
if we minimize (\ref{eq:loss}) simultaneously in terms of
$\{\bsb_j\}_{j=1}^C$ and $\bbeta$. We refer to this as ``clubbing'',
following \cite{ghos:hast:owen:2021}. We also combine this technique
with our variational EM approach in Section~\ref{sec:vari-em-cross} to achieve fast convergence.
\subsection{An efficient way to implement clubbing}
\label{sec.clubbing}
We consider the case of estimating  $\bbeta$ and $\bsa_i$ with
$\{\bsb_j\}_{j=1}^C$ fixed.
This  would be equivalent to minimizing the loss given by (\ref{eq:loss}) keeping $\{\bsb_j\}$ fixed.
Thus, if $r_{ij} = y_{ij} - \mathbf{x}^\tran_{ij} \bsb_j $, then we 
\begin{equation}
\minimize_{\bbeta,\{\bsa_i\}}\sum_{i=1}^R\sum_{j=1}^{C}z_{ij}\big(r_{ij} - \mathbf{x}_{ij} ^{\tran}(\bbeta + \bsa_i)\big)^2+\sum _{i=1}^R \bsa_i^{\tran}\Lambda \bsa_i ,
\end{equation}
which would be equivalent to simultaneously solving the following equations:
\begin{equation}\label{eq:a_i_clubbed}
\hat {\bsa_i }= \bigg(\sum_{j=1}^{C} z_{ij}\mathbf{x}_{ij} \mathbf{x}^{\tran}_{ij}+ \Lambda\bigg)^{-1}\bigg(\sum_{j=1}^{C} \Big(r_{ij} - \mathbf{x}_{ij} ^{\tran}\bbeta\Big) z_{ij}\mathbf{x}_{ij} \bigg),
\end{equation}
\begin{equation}
\label{eq:beta_clubbed}
\hat\bbeta = \bigg(\sum_{j=1}^C \sum_{i=1}^{R}z_{ij} \mathbf{x}_{ij} \mathbf{x}^{\tran}_{ij}\bigg)^{-1}\bigg( \sum_{i=1}^{R} \sum_{j=1}^C \Big(r_{ij} - \mathbf{x}_{ij} ^{\tran}\bsa_i\Big)z_{ij}\mathbf{x}_{ij} \bigg),
\end{equation}
where $\Lambda = \sigma^2_e \Sigma_a^{-1}$.\\
Let $f_a$ denote the fit ($N$ vector) corresponding to the terms $\{\mathbf{x}_{ij} ^{\tran} \hat {\bsa_i}\}$ and $X \hat\bbeta$ be the $N$-vector of fits for the fixed effects. Then from $(\ref{eq:a_i_clubbed})$, we can write
\begin{equation}
\label{eq:sim_ai}
f_a = S_a (r - X \hat \bbeta),
\end{equation}
and from $(\ref{eq:beta_clubbed})$, we have \begin{equation}
\label{eq:sim_beta}
X^{\tran}X \hat \bbeta = X^{\tran} (r - f_a).\end{equation}
Here, $S_a$ is the implicit $N \times N$ linear operator that computes
these fits, obtained by solving $(\ref{eq:a_i_clubbed})$ for each
$i$. We note that if the observations were ordered by groups of $i$, then
  $S_a$ would be block-diagonal, with the $i$th block a generalized ridge
  operator matrix for producing the fits for the observations with
  that value of $i$. This also implies that $S_a$ is always a symmetric matrix.  Plugging $(\ref{eq:sim_ai})$ into $(\ref{eq:sim_beta})$, we have
\begin{equation}
X^{\tran}X \hat \bbeta = X^{\tran} (r - S_a (r - X \hat \bbeta)).
\end{equation}
Collecting terms, we have
\begin{equation}
X^{\tran} \left(I - S_a\right) X \hat \bbeta = X^{\tran} \left(I - S_a\right) r,
\end{equation}
and hence
\begin{equation}
\hat \bbeta = \left[ X^{\tran} \left(I - S_a\right) X\right] ^{-1} X^{\tran} \left(I - S_a\right) r.
\end{equation}
This means that we need to apply $S_a$ to each of the columns of X, compute the fits, and take residuals, to produce $\tilde{X} = (I - S_a) X$, then
\begin{equation}
\hat \bbeta = \left( \tilde{X}^{\tran} X \right)^{-1} \tilde{X}^{\tran} r \mbox{ ( as $S_a$ is symmetric} ).
\end{equation}
We plug this $\hat\bbeta$ in equation $(\ref{eq:a_i_clubbed})$ to get
$\hat\bsa_i$ for each $i$. Now we repeat this process holding the $\{\bsa_i\}$ fixed at their latest
values $\{\hat {\bsa_i}\}$, and  estimate $\bbeta$ and $\{\bsb_j\}$ together. 
We continue the above process till fits of the terms, $\mathbf{x}_{ij} ^{\tran}
\bbeta, \{\mathbf{x}_{ij} ^{\tran} \bsa_i \}$ and $ \{\mathbf{x}_{ij} ^{\tran} \bsb_j\}$
converge.

 The time complexity for each step of backfitting with or without clubbing is $\mathcal{O}(Np^2)$.
\section{Estimation of covariance parameters using the
method of moments}\label{sec.covariance}
Method of moments (MoM) is a widely used approach for estimating
covariance parameters, but for the most part has been limited to cases
where there are a small number of scalar parameters. See, for example,
\cite{wu2012efficient} and \cite{crevarcomp}. In our setting, we have a
novel MoM approach   for crossed
random effects with covariance matrices of any  dimension.
Our MoM approach has
algorithmic complexity linear in $N$, and hence is scalable.

We discuss the estimation of covariance parameters under two scenarios: when $\Sigma_a$ and $\Sigma_b$ are assumed to be diagonal, and when nothing is assumed.
\subsection{\texorpdfstring{$\Sigma_a$}{Sigmaa} and \texorpdfstring{$\Sigma_b$}{Sigmab} are diagonal  matrices.}\label{sec:sigma_a-sigma_b-are}
If
$\Sigma_a$ and $\Sigma_b$ are diagonal matrices, we can obtain consistent estimates of $\Sigma_a, \Sigma_b$, and $\sigma^2_e$ using the method of moments. Our process involves solving equations of the form
\begin{equation}
\label{eq:mm1}
\begin{split}
\Ex\bigg[\sum_{i=1}^{R} \sum_{j=1}^{C} z_{ij} r^2_{ij}-\sum_{i=1}^{R}\frac{\left(\sum_{j=1}^{C} z_{ij} r_{ij}x^{(s)}_{ij}\right)^2}{\sum_{j=1}^{C} z_{ij}x^2_{ij(s)}}\bigg{|}X\bigg] & = \sum_{i=1}^{R} \sum_{j=1}^{C} z_{ij}\hat r^2_{ij}\\
& \hspace{0.5cm} -\sum_{i=1}^{R}\frac{\left(\sum_{j=1}^{C} z_{ij}\hat r_{ij}x^{(s)}_{ij}\right)^2}{\sum_{j=1}^{C} z_{ij}x^2_{ij(s)}},
\end{split}
\end{equation}
\begin{equation}
\label{eq:mm2}
\begin{split}
\Ex\bigg[\sum_{j=1}^{C} \sum_{i=1}^{R} z_{ij} r^2_{ij}-\sum_{j=1}^{C}\frac{\left(\sum_{i=1}^{R} z_{ij} r_{ij}x^{(s)}_{ij}\right)^2}{\sum_{i=1}^{R} z_{ij}x^2_{ij(s)}}\bigg{|}X\bigg] & = \sum_{j=1}^{C} \sum_{i=1}^{R} z_{ij}\hat r^2_{ij}\\
& \hspace{0.5cm} -\sum_{j=1}^{C}\frac{\left(\sum_{i=1}^{R} z_{ij}\hat r_{ij}x^{(s)}_{ij}\right)^2}{\sum_{i=1}^{R} z_{ij}x^2_{ij(s)}},
\end{split}
\end{equation}
for all s $\in \{0, \cdots, p\}$, and
\begin{equation}
\label{eq:mm3}
\Ex\bigg[\sum_{i=1}^{R} \sum_{j=1}^C z_{ij} \left( r_{ij} - \bar{r}_{..} \right)^2 \bigg{|}X\bigg]=\sum_{i=1}^{R} \sum_{j=1}^C z_{ij} \left( \hat r_{ij} - \bar{\hat r}_{..} \right)^2 .
\end{equation}
\begin{remark}
Here, $x^{(s)}_{ij}$ represents s-th covariate of the $(i,j)$
client-item pair, $r = Y - X \bbeta$, and $\hat r = Y - X
\hat\bbeta_{\ols}$.
\end{remark}

These equations can be seen as a generalization of the moment
equations considered by~\cite{ghos:hast:owen:2021}; we give more
insight into these equations in Appendix~\ref{appendix.mom}.
Using the matrix formulation in the appendix, we can show that
\begin{equation}
\label{eq:simplify}
\begin{split}
\Ex\bigg[&\sum_{i=1}^{R} \sum_{j=1}^{C} z_{ij} r^2_{ij}-\sum_{i=1}^{R}\frac{\left(\sum_{j=1}^{C} z_{ij} r_{ij}x^{(s)}_{ij}\right)^2}{\sum_{j=1}^{C} z_{ij}x^2_{ij(s)}} \bigg{|} X\bigg] \\
&=\sum_{s'=0}^{p} \sum_{i=1}^{R}\left[ \sum_{j=1}^{C} z_{ij} x^{(s')^2}_{ij}-\frac{\left(\sum_{j=1}^{C} z_{ij} x^{(s)}_{ij} x^{(s')}_{ij} \right)^2}{\sum_{j=1}^{C} z_{ij} x^{(s)^2}_{ij}}\right]\Sigma_{a,s's'} \\
& \hspace{0.8cm}+\sum_{s'=0}^{p}\sum_{i=1}^R\left[ \sum_{j=1}^{C}z_{ij} x^{(s')^2}_{ij} -\frac{\sum_{j=1}^{C}z_{ij} x^{(s)^2}_{ij} x^{(s')^2}_{ij}}{\sum_{j=1}^{C} z_{ij} x^{(s)^2}_{ij}}\right]\Sigma_{b,s's'} + \sigma^2_e (n-R) .\\
\end{split}
\end{equation}
\begin{remark}
The proof is provided in Appendix~\ref{appendix.mom}. The complexity
to compute all the quantities above is $\mathcal{O}(Np^2)$. Once these
quantities are computed, we solve equations $(\ref{eq:mm1}),
(\ref{eq:mm2})$ and $(\ref{eq:mm3})$ to get the covariance
parameters. Thus we obtain $2p+3$ equations involving $2p+3$ unknown
variance parameters, which we can solve to obtain MoM estimates for
the covariance parameters.
\end{remark}

Under the regularity conditions mentions in the theorems below, the covariance parameters estimated using the suggested MoM approach are consistent. We have implemented the MoM approach for the case when the covariance matrices are diagonal in Section $\ref{subsec.diagonal}$. The approach is found to provide empirically consistent estimates for the covariance parameters which are then used to estimate the fixed parameter using backfitting.
\begin{restatable}{theorem}{betaols}
\label{betaolsconsistent}
Let $X_{r}^{(a)} (X_{c}^{(b)})$ be the observations arranged in row (column) blocks such that each block has a fixed level for the first (second) factor. Also assume that $\max_{1\leq r \leq R}\frac{\lambda_{\max}(X_{r}^{(a)}\Sigma_{a}X_{r}^{(a) \tran})}{N} \to 0$, $\max_{1\leq c \leq C}\frac{\lambda_{\max}(X_{c}^{(b)}\Sigma_{b}X_{c}^{(b) \tran})}{N} \to 0$ and $\lambda_{min}(X^{\tran}X)/N \geq c$ for some $c > 0$, as $N \to \infty.$ Then $\hat{\bbeta}_{\ols}$ is consistent.
\end{restatable}
\begin{restatable}{theorem}{mmconsis}
The method of moment estimates obtained using the above method is consistent if the fourth moments for $\mathbf{x}_{ij} , i \in \{1, \cdots, R\}, \, j \in \{1, \cdots, C\} $ are uniformly bounded and $\hat{\bbeta}_{\ols}$ is consistent.
\end{restatable}
\subsection{ \texorpdfstring{$\Sigma_a$}{Sigmaa} and \texorpdfstring{$\Sigma_b$}{Sigmab} are unstructured}
If there are no constraints on the structure of $\Sigma_a$ and
$\Sigma_b$, the problem of estimating covariance parameters becomes
more complicated. In this case we would solve additional equations of the form
\begin{equation} \hat r^{\tran}(I-Q_r)(I-Q_s)\hat r = \Ex[r^{\tran}(I-Q_r)(I-Q_s)r]\end{equation}
for each $r, s \in \{0, \cdots, p\}$.

However, we need $\Sigma_a$ and $\Sigma_b$ to be positive
definite, and solving linear equations for $\Sigma_a$ and $\Sigma_b$
may not necessarily guarantee this without imposing additional constraints.
We cannot rely on consistency either, since we would need an
extremely large sample to guarantee a positive definite solution.

Rather than pursuing this approach, we developed a likelihood method using variational EM, which we discuss next.
\section{Likelihood approaches for estimating all the parameters}\label{sec.variational}
\subsection{Maximum likelihood estimation}
Here we attempt to maximize the likelihood function for $(\bbeta,
\Sigma_a, \Sigma_b, \sigma^2_e)$ in (\ref{eq:likelihood}),
\begin{equation}
\label{eq:mle}
\maximize_{\bbeta, \Sigma_a, \Sigma_b, \sigma^2_e} \frac{1}{(\sqrt{2\pi)^N|\mathcal{V}|}} \exp\left(-\frac{1}{2}(Y - X \bbeta )^{\tran} \mathcal{V}^{-1}(Y-X \bbeta)\right),
\end{equation}
with $\mathcal{V}$ defined as in (\ref{eq:V}). We showed in Section~\ref{sec.backfitting} that if the covariance parameters are known, we
can use backfitting to obtain $\hat\bbeta_{\mathrm{MLE}}$ in an
efficient way. The problem becomes much more complicated if the
covariance parameters are unknown, since the parameters $\Sigma_a$,
$\Sigma_b$, $\sigma_e^2$ are buried in $\mathcal{V}$ in a complex manner. The R package ``lme4'' \citep{lme4}  implements Gaussian
maximum likelihood estimation for the random effects model, but it is
computationally expensive for large problems. We will compare the
efficiency of our proposed algorithm to ``lme4'' in Section~\ref{sec:results}.
\subsection{Expectation Maximization Algorithm}
One potential solution to simplify the maximum likelihood problem
analytically is to use the EM algorithm. In our setup, EM would treat
the $\{\bsa_i\}_{i=1}^{R}$ and $\{\bsb_j\}_{j=1}^{C}$ as unobserved/missing
quantities. In this case, the complete log-likelihood is (ignoring constants)
\begin{equation}
\label{eq: complete_likelihood}
\begin{split}
l(Y;&\{\bsa_i\} ;\{\bsb_j\}) = - \frac{1}{2} \left(N \log(\sigma^2_e) + R \log (|\Sigma_a|)+ C \log (|\Sigma_b|)\right)\\
- \frac{1}{2} & \left(\sum_{i=1}^R  \sum_{j=1}^{C}\frac{z_{ij}(y_{ij} - \mathbf{x}_{ij} ^{\tran}(\bbeta + \bsa_i +\bsb_j))^2}{\sigma^2_e} + \sum_{i=1}^R \bsa_i^{\tran} \Sigma^{-1}_a \bsa_i + \sum_{j=1}^C \bsb_j^{\tran} \Sigma_b^{-1} \bsb_j\right).
\end{split}
\end{equation}
We have seen part of this expression earlier when we discussed
backfitting to obtain $\hat\bbeta_{\gls}$ in Section
\ref{thm:robinson}. \\
\\
EM is an iterative algorithm that involves two steps:
\begin{description}
\item[\bf E-step:] Compute the expected value of the complete
  log-likelihood conditional on $Y$ and the current parameters
  $\hat\btheta^{(k)}$, giving the function   $\ell_E(\btheta)$. 
\item[\bf M-step:] Maximize  $\ell_E(\btheta)$ to produce new parameter estimates, $\hat{\btheta}^{(k+1)}$.
\end{description}
We have used $\btheta$ to represent all the parameters $(\beta,
\Sigma_a, \Sigma_b, \sigma^2_e)$. The  E-step would require computing $\Ex[ \bsa_i \bsa^{\tran}_{i}| Y;\btheta]$, $\Ex[ \bsb_j \bsb^{\tran}_{j}| Y;\btheta]$, and
\begin{equation}
\label{eq:em_residual}
\Ex\left[\sum_{i=1}^R \sum_{j=1}^C z_{ij}\left(y_{ij}- \mathbf{x}^\tran_{ij}(\bbeta^{(k)} + \bsa_i + \bsb_j)\right)^2|Y;\btheta\right].
\end{equation}
Computing (\ref{eq:em_residual}) involves computing $\Ex[\bsa_i |
Y;\btheta],$ $\Ex[\bsb_j | Y;\btheta]$, and $\Ex[\bsa^{\tran}_i \bsb_j |
Y;\btheta]$ for $i \in \{0, \cdots, p\}$ and $ j \in \{0, \cdots,
p\}$. All is manageable except the terms $\Ex[\bsa^{\tran}_i \bsb_j |
Y;\btheta]$. Here, $(\bsa_1,\cdots, \bsa_R, \bsb_1,\cdots \bsb_C)$ follows a
multivariate normal distribution conditional on $Y$ and the covariance
matrix for the conditional distribution would only be obtained by
inverting a matrix of order $(R+C)p \times (R+C)p$.
The estimate of  $\sigma^2_e$ in the M-step requires the
computation(\ref{eq:em_residual}), and presents us with time
complexity of at least $\{(R+C)p\}^3$ which is greater than or equal to
$N^{3/2} p^3$.

Variational EM helps us out here. Instead of using the
conditional distribution detailed above, we consider conditional distributions for $(\bsa_1,\cdots, \bsa_R, \bsb_1,\cdots \bsb_C)$ under which $\{\bsa_i\}_{i=1}^R$ are independent of $\{\bsb_j\}_{j=1}^C$.
\subsection{Variational EM Algorithm}
Variational EM was introduced by~\cite{beal:zoubin:2003} and is
applicable when the missing data has a graphical model structure as in our
case, i.e., $\bsa_i$ and $\bsb_j$ are correlated conditional on
$Y$.  \cite{menictas2022streamlined} used streamlined
variational inference for the Bayesian crossed random effects
model. Our purpose in using variational EM is to estimate an
approximation to the maximum likelihood estimate in the frequentist
paradigm. The algorithm we develop here is a substitute for the
maximum likelihood estimate for the crossed random effects
model. Since we are unable to directly estimate the covariance parameters
efficiently, variational EM is a practical alternative.
Though variational EM does not come with the
guarantee of maximization of the likelihood, the estimates are found
to be empirically consistent in our simulations. Variational EM also
has a close connection with backfitting which we discuss further.
\subsubsection{Review of variational EM}

Let $\balpha$ denote the vector of unobserved quantities, $\btheta$
denote the set of parameters, and $l(\btheta)$ represent the likelihood
function of $\btheta$. In the E-step, we calculate the expectation of
the complete likelihood given $Y$, i.e., sufficient statistics of
$\balpha$ get replaced by their expectation under the distribution
$\Prob(\balpha|Y)$. Under factored variational EM, instead of finding
the expectation of the complete likelihood given $Y$, we substitute the conditional distribution of $\balpha$ given $Y$ with $q(\balpha)$ belonging in a restricted set $Q$ of probability distributions.\\
We define
\begin{equation} F(q,\btheta) : = \int q(\balpha)\log \frac{\Prob(Y,\balpha|\btheta)}{q(\balpha)} d\balpha \end{equation}
which we view as an approximation to $l(\btheta)$; if
  $q(\balpha)=\Prob(\balpha|Y,\btheta)$ then  $F(q,\btheta)=l(\btheta)$.
In more detail,
\begin{align*}
F(q,\btheta) & = \int q(\balpha)\log \frac{\Prob(Y,\balpha|\btheta)}{q(\balpha)} d \balpha\\
& = \int q(\balpha)\log \frac{ \Prob(Y|\btheta)\Prob(\balpha|Y,\btheta)}{q(\balpha)} d \balpha\\
& = \int q(\balpha) \left\{\log \Prob(Y|\btheta) + \log \left( \frac{\Prob(\balpha|Y,\btheta)}{q(\balpha)}\right) \right\} d \balpha\\
& = l(\btheta) - \mathrm{KL} \left(q(\balpha)\| \Prob(\balpha|Y,\btheta)\right).
\end{align*}
As, $\mathrm{KL} \left(q(\balpha)\| \Prob(\balpha|Y,\btheta)\right) \geq 0$, $F(q,\btheta) \leq l(\btheta)$ with equality iff $q(\balpha) = \Prob(\balpha|Y,\btheta)$. At each variational E-step, we maximize $F(q,\btheta) $ with respect to $q$ within the set $Q$. This implies that variational EM maximizes a lower bound to the likelihood.
If $Q = \{q: q(\balpha) = \prod q_u(\balpha_u)\}$, it follows from~\cite{beal:zoubin:2003} that for the exponential family, optimal update for distribution of each component of $\balpha$ is given by
\begin{equation}
q_u^{(k)}(\balpha_u) \propto \Prob \left(Y, \balpha_u| \balpha_{-u} = \Ex_{q^{(k-1)}} (\balpha_{-u})\right).
\end{equation}
Variational EM possesses the property that
\begin{equation}F(q^{(k-1)},\btheta^{(k-1)}) \leq F(q^{(k)},\btheta^{(k-1)}) \leq F(q^{(k)},\btheta^{(k)}) . \end{equation}
\subsubsection{Variational 	EM for the crossed random effects model with random slopes}\label{sec:vari-em-cross}
For the crossed random effects model (\ref{eq:random_slopes}), $\balpha$ would have two components, $\{\bsa_i\}_{i=1}^R$ and $\{\bsb_j\}_{j=1}^C$. Here, $\btheta$ represents the set of parameters ($\bbeta, \Sigma_a, \Sigma_b, \sigma^2_e$). We choose
\begin{equation}Q = \left\{ q: q(\bsa_1,\cdots, \bsa_R, \bsb_1,\cdots , \bsb_C) = q_a(\bsa_1, \cdots, \bsa_R) q_b(\bsb_1, \cdots, \bsb_C)\right\} .\end{equation}
\begin{restatable}{lemma}{varE}
\label{lemma:var_E}
The optimal distributions for $\{\bsa_i\}_{i=1}^R$ and $\{\bsb_j\}_{i=1}^R$ in the variational E-step are given by \[ q^{(k)}_a(\bsa_1, \cdots, \bsa_R) = \prod_{i=1}^R q^{(k)}(\bsa_i) \quad \mbox{ and } \quad q^{(k)}_b(\bsb_1, \cdots, \bsb_C) = \prod_{j=1}^C q^{(k)}(\bsb_j) \]
with
\begin{equation}q^{(k)}(\bsa_i) = N(\mu^{(k)}_{a,i}, \Sigma^{(k)}_{a,i}) \quad \mbox{and} \quad q^{(k)}(\bsb_j) = N(\mu^{(k)}_{b,j},\Sigma^{(k)}_{b,j}) , \end{equation}
where
\begin{equation}
\label{eq:var_mu_a}
\mu^{(k)}_{a,i} = \Sigma^{(k)}_{a,i} \left(\frac{\sum_{j=1}^{C}z_{ij}\big(y_{ij} - \mathbf{x}_{ij} ^{\tran}\big(\bbeta^{(k-1)} +\mu^{(k-1)}_{b,j}\big)\big)\mathbf{x}_{ij} }{\left(\sigma^{(k-1)}_e\right)^{2}}\right)
\end{equation}
and
\begin{equation}
\label{eq:var_sigma_a}
\Sigma^{(k)}_{a,i} = \left((\Sigma^{(k-1)}_a)^{-1} +\frac{\sum_{j=1}^{C}z_{ij}\mathbf{x}_{ij}  \mathbf{x}_{ij} ^{\tran}}{\left(\sigma^{(k-1)}_e\right)^2}\right)^{-1}
\end{equation}
with $\mu^{(k)}_{b,j}$ and $\Sigma^{(k)}_{b,j}$ having similar definitions.
\end{restatable}
\begin{proof}
The proof can be found in Appendix~\ref{sec:proof-lemma-refl}.
\end{proof}

It is important to note that while updating the distribution of
$\{\bsa_i\}_{i=1}^R$, we use the parameters of the distribution of
$\{\bsb_j\}_{j=1}^C$ updated in the previous step. Likewise when updating the distribution of $\{\bsb_j\}_{j=1}^C$, we use parameters of the distribution of $\{\bsa_i\}_{i=1}^R$ recently updated. Thus, the mechanism of the update is very similar to the backfitting approach where we iterate through updates in $\{\bsa_i\}_{i=1}^R$ and $\{\bsb_j\}_{j=1}^C$.
Each of the quantities in (\ref{eq:var_mu_a}) and (\ref{eq:var_sigma_a}) can be computed in $\mathcal{O}(N)$ as computing $\mu^{(k)}_{a,i}$ and $\mu^{(k)}_{b,j}$ is equivalent to obtaining the estimates of $\{\bsa_i\}_{i=1}^{R}$ and $\{\bsb_j\}_{j=1}^{C}$ in each step of backfitting. \\
\begin{description}
\item[\bf Variational E step:] The expectation of the complete likelihood (\ref{eq: complete_likelihood}) can be computed with respect to the probability distribution $q^{(k)}$, though we can skip to the M-step and fill in the details there.\\\
\item[\bf M step:] The estimate
$\Sigma_a^{(k)}$ is given by
\begin{equation}
\label{eq:Sigma_a_k}
\begin{split}
\Sigma_a^{(k)} : & = \frac{1}{R}\Ex_{q^{(k)}}\left[ \sum_{i=1}^R \bsa_i \bsa_i^{\tran}\right] \\
& = \frac{1}{R}\sum_{i=1}^R \left(\mu^{(k)}_{a,i} {\mu^{(k)}_{a,i}}^{\tran} + \Sigma^{(k)}_{a,i}\right)
\end{split}
\end{equation}
Similarly,
\begin{equation}
\label{eq:Sigma_b_k}
\Sigma_b^{(k)} = \frac{1 }{C}\sum_{j=1}^C\left(\mu^{(k)}_{b,j} {\mu^{(k)}_{b,j}}^{\tran} + \Sigma^{(k)}_{b,j}\right). \end{equation}
Also,
\begin{equation}
\label{eq:beta_k}
\bbeta^{(k)} = \bigg(\sum_{i=1}^R\sum_{j=1}^{C}z_{ij}\mathbf{x}_{ij}  \mathbf{x}_{ij} ^{\tran}\bigg)^{-1}\sum_{i=1}^{R}\sum_{j=1}^{C}z_{ij}\mathbf{x}_{ij} \Ex\left(y_{ij} - \mathbf{x}_{ij} ^{\tran}\left(\mu^{(k)}_{a,i} + \mu^{(k)}_{b,j}\right)\right), \end{equation}
and finally
\begin{align*}
\left(\sigma^{(k)}_e\right)^2 & = \frac{1}{N}\sum_{i=1}^R \sum_{j=1}^C z_{ij} \Ex_{q^{(k)}}\left[ \left(y_{ij}- \mathbf{x}^\tran_{ij}(\bbeta^{(k)} + \bsa_i + \bsb_j)\right)^2\right]\\
& = \frac{1}{N}\sum_{i=1}^R \sum_{j=1}^C z_{ij}\left(\Ex_{q^{(k)}}\left[\left(y_{ij}- \mathbf{x}^\tran_{ij}(\bbeta^{(k)} + \bsa_i + \bsb_j)\right)\right]\right)^2\\
& \hspace{1.5cm}+ \frac{1}{N}\sum_{i=1}^R \sum_{j=1}^C z_{ij} \textit{Var}_{q^{(k)}}\left[\left(y_{ij}- \mathbf{x}^\tran_{ij}(\bbeta^{(k)} + \bsa_i + \bsb_j)\right)\right]\\
& = \frac{1}{N}\sum_{i=1}^R \sum_{j=1}^C z_{ij}\left(\Ex_{q^{(k)}}\left[\left(y_{ij}- \mathbf{x}^\tran_{ij}(\bbeta^{(k)} + \bsa_i + \bsb_j)\right)\right]\right)^2\\
& \hspace{1.5cm}+ \frac{1}{N}\sum_{i=1}^R \sum_{j=1}^C z_{ij} \textit{Var}_{q^{(k)}}\left[ \mathbf{x}^\tran_{ij}( \bsa_i + \bsb_j)\right]\\
& = \frac{1}{N}\sum_{i=1}^R \sum_{j=1}^C z_{ij}\left(y_{ij}- \mathbf{x}^\tran_{ij}(\bbeta^{(k)} + \mu^{(k)}_{a,i} + \mu^{(k)}_{b,j})\right)^2\\
&\hspace{4cm} + \frac{1}{N}\sum_{i=1}^R \sum_{j=1}^C z_{ij} \mathbf{x}^\tran_{ij} \left(\Sigma^{(k)}_{a,i} + \Sigma^{(k)}_{b,j}\right)\mathbf{x}_{ij} .
\end{align*}
\end{description}
The computational cost of the variational E-step and M-step is $\mathcal{O}(Np^2)$.
\subsection{Connection of variational EM with backfitting}
If the covariance parameters were known or otherwise previously estimated, the goal of variational EM would be just to estimate $\bbeta$. In that case, optimal distribution for $\{\bsa_i\}_{i=1}^R$ and $\{\bsb_j\}_{i=1}^R$ in the variational E-step are given by \begin{equation}q^{(k)}(\bsa_i) = N(\mu^{(k)}_{a,i}, \Sigma_{a,i}) \quad \mbox{and} \quad q^{(k)}(\bsb_j) = N(\mu^{(k)}_{b,j},\Sigma_{b,j}), \end{equation}
where
\begin{equation}
\label{eq:var_mu_a_fixed}
\mu^{(k)}_{a,i} = \Sigma_{a,i} \left(\frac{\sum_{j=1}^{C}z_{ij}\big(y_{ij} - \mathbf{x}_{ij} ^{\tran}\big(\bbeta +\mu^{(k-1)}_{b,j}\big)\big)\mathbf{x}_{ij} }{\sigma_e^{2}}\right)
\end{equation}
and
\begin{equation}
\label{eq:var_sigma_a_fixed}
\Sigma_{a,i} = \left(\Sigma_a^{-1} +\frac{\sum_{j=1}^{C}z_{ij}\mathbf{x}_{ij}  \mathbf{x}_{ij} ^{\tran}}{\sigma_e^2}\right)^{-1}
\end{equation}
with $\mu^{(k)}_{b,j}$ and $\Sigma_{b,j}$ having similar definitions. It is important to note that the variance of $\{\bsa_i\}$ and $\{\bsb_j\}$ are now fixed over the iterations, because the covariance parameters, $\Sigma_a$, $\Sigma_b$, and $\sigma^2_e$ are assumed to be known and fixed quantities. Equations (\ref{eq:var_mu_a_fixed}) and (\ref{eq:var_sigma_a_fixed}) imply that
\begin{equation}
\label{eq:a_i_em}
\mu^{(k)}_{a,i} = \left(\sigma^2_e\Sigma_a^{-1} +\sum_{j=1}^{C}z_{ij}\mathbf{x}_{ij}  \mathbf{x}_{ij} ^{\tran}\right)^{-1} \left(\sum_{j=1}^{C}z_{ij}\big(y_{ij} - \mathbf{x}_{ij} ^{\tran}\big(\bbeta +\mu^{(k-1)}_{b,j}\big)\big)\mathbf{x}_{ij} \right),
\end{equation}
and similarly,
\begin{equation}
\label{eq:b_j_em}
\mu^{(k)}_{b,j} = \left(\sigma^2_e\Sigma_b^{-1} +\sum_{i=1}^{R}z_{ij}\mathbf{x}_{ij}  \mathbf{x}_{ij} ^{\tran}\right)^{-1} \left(\sum_{i=1}^{R}z_{ij}\big(y_{ij} - \mathbf{x}_{ij} ^{\tran}\big(\bbeta +\mu^{(k)}_{a,i}\big)\big)\mathbf{x}_{ij} \right).
\end{equation}
The estimate of $\bbeta$ in E-step remains unchanged and is given by
\begin{equation}
\label{eq:beta_em}
\bbeta^{(k)} = \bigg(\sum_{i=1}^R\sum_{j=1}^{C}z_{ij}\mathbf{x}_{ij}  \mathbf{x}_{ij} ^{\tran}\bigg)^{-1}\sum_{i=1}^{R}\sum_{j=1}^{C}z_{ij}\mathbf{x}_{ij} \Ex\left(y_{ij} - \mathbf{x}_{ij} ^{\tran}\left(\mu^{(k)}_{a,i} + \mu^{(k)}_{b,j}\right)\right). \end{equation}

It is easy to see that update of $\mu^{(k)}_{a,i}$, $\mu^{(k)}_{b,j}$, and $\hat \bbeta$ in each step of variational EM (\ref{eq:a_i_em}) , (\ref{eq:b_j_em}) with known covariance parameters is same as update of $\hat \bsa_{i}$, $\hat\bsb_j$ and $\bbeta$ in each step of backfitting . This draws a close connection between variational EM and backfitting. We know that $\hat\bsa_i \to \Ex[\bsa_i|Y]$ and $\hat\bsb_j \to \Ex[\bsb_j|Y]$ over the iterations in backfitting which provides a plausible justification for variational EM.

 \subsection{Estimation of covariance of \texorpdfstring{$\hat{\bbeta}$}{beta}}
Let $S_G$ denotes the linear operator corresponding to the fit due to random components $\{\bsa_i\}$ and $\{\bsb_j\}$, i.e., $S_G(Y)_{ij} = \mathbf{x}_{ij} ^\tran (\hat \bsa_i +\hat \bsb_j)$, then similar to the development in Section $\ref{sec.clubbing}$, we can write
\begin{equation}\hat{\bbeta} = (X^{\tran} (I- S_G) X)^{-1} (X^{\tran} (I-S_G) Y). \end{equation}

It is important to note here that $S_G$ is a \emph{symmetric} operator as it is an operator matrix for a generalised linear regression. Hence we can write
\begin{equation}\hat{\bbeta} = \left(X^{\tran} (X - S_G X)\right)^{-1} \left((X - S_G X)^\tran Y\right). \end{equation}
Another important thing to mention here is that we never compute or store the operator $S_G$, instead we apply the backfitting procedure similar to Section $\ref{sec.clubbing}$ to obtain an estimate of $S_G(X_j)$ given by $\tilde{S}_G(X_j)$ for each column $X_j$ of $X$.

Let $\tilde{X} = (I - \tilde{S}_G) X$, then our estimate for the covariance matrix of $\hat{\bbeta}$ is
\begin{equation}(X^{\tran} \tilde{X})^{-1} \tilde{X}^{\tran} \widehat{\mathcal{V}} \tilde{X} (\tilde X^{\tran} X)^{-1}, \end{equation}
where $\widehat{\mathcal{V}}$ is an estimate for the covariance matrix of $Y$.  Using the definition of $\widehat{\mathcal{V}}$ in $(\ref{eq:V}$), $\cov(\hat{\bbeta})$ is given by
\begin{equation}
\label{eq:sd_beta_hat}
(X^{\tran} \tilde{X})^{-1} \tilde{X}^{\tran} \left( X \hat\Sigma_a X^\tran \bullet Z_a Z_a^\tran + X \hat\Sigma_b X^\tran \bullet Z_b Z_b^\tran  + \hat \sigma^2_e I_N \right)  \tilde{X}  (\tilde X^{\tran} X)^{-1}. \end{equation}

If we directly implement the matrix multiplication above, the operation cannot be performed in time complexity which is linear in $N$. However, exploiting the structure of the expression, a careful implementation brings the complexity down to linear in $N$. The details can be found in Appendix~\ref{sec.covariance_of_beta_hat}. The covariance matrix of $\hat{\bbeta}$ can be used to perform inference on $\bbeta$ and z-tests.
\section{Results}\label{sec:results}
We compare the time required for maximum likelihood estimation
by ``lmer'' \cite{lme4} with our methods. We also illustrate the empirical consistency of our methods using simulations. Finally, we implement our method on MovieLens data and Stitch Fix data, and compare the results with the OLS fit and the crossed random effects model with random intercepts implemented by ~\cite{ghos:hast:owen:2021}.
\subsection{Simulations}
\subsubsection{Error in estimating covariance matrices}\label{error_sigma_a_b}
To articulate the errors in estimating $\Sigma_a$ and $\Sigma_b$, we are using two metrics.
\begin{itemize}
\item \textbf{KL divergence}:
We define the distance between $\Sigma_a$ and $\hat \Sigma_a$ in the following way. Let $\Prob_{1}$ denote the normal distribution with mean zero and covariance $\Sigma_a$ and $\Prob_{2}$ denote the normal distribution with mean zero and covariance $\wh\Sigma_a$, then we define the distance by
\begin{equation}\Ex_{\Prob_1}\bigg[\log \bigg(\frac{d\Prob_{1}}{d\Prob_{2}}\bigg) \bigg] = \frac{1}{2}\Bigg[ \tr\bigg(\wh\Sigma_a^{-1}\Sigma_a\bigg) -\log(|\wh\Sigma_a^{-1}\Sigma_a|) -d\Bigg]. \end{equation}
\item \textbf{Frobenius norm of difference between $\Sigma_a$ and $\wh{\Sigma}_a$}:
\[\|\Sigma_a-\wh{\Sigma}_a\|_F\]
\end{itemize}
For the simulations, we consider the model given by
  (\ref{eq:random_slopes}) with random slopes for all the covariates
  for clients as well as items. However, as mentioned in \cref{remark:randomslopes}, in practice we typically
  include  random slopes for only a subset of covariates for both clients and items. We present the results for the general setup in Appendix~\ref{sec.generalized_results}.  

\subsubsection{Diagonal case}\label{subsec.diagonal}
For our simulations, we choose $p = 3, \beta_0 = 0.1$, $\bbeta = (0.2, 0.3, 0.4)$, and $\sigma^2_e = 1$. We considered \[\Sigma_a = \begin{bmatrix}
0.3 & 0 & 0 & 0\\
0 & 0.3 & 0 & 0\\
0 & 0 & 0.3 & 0\\
0 & 0 & 0 & 0.3\\
\end{bmatrix}
\mbox{  and  }\Sigma_b = \begin{bmatrix}
0.1 & 0 & 0 & 0\\
0 & 0.1 & 0 & 0\\
0 & 0 & 0.1 & 0\\
0 & 0 & 0 & 0.1\\
\end{bmatrix}.\]
We choose $R = C = N^{0.6}$ for our simulations which satisfies the
sparsity assumption as seen in practice/e-commerce data. For each
value of $N$, we perform 100 simulations and the plot the average
values.
We have proposed the following algorithms:
\begin{description}
\item \textbf{Backfitting}: We estimate covariance parameters using the method of moments (MoM) and then perform backfitting to estimate the fixed parameter.
\item \textbf{Clubbed Backfitting}: We estimate covariance parameters using MoM and then perform backfitting with clubbing to estimate the fixed parameter.
\item \textbf{Variational EM}: We use MoM to get an initial estimate of covariance parameters and then we improve the estimates in each iteration using the variational M-step along with estimating fixed parameters.
\item \textbf{Clubbed variational EM}: Similar to variational EM, we start with MoM estimates for covariance parameters and improvise them using variational EM. We apply ``clubbing" trick while estimating $\bbeta$ to fasten the convergence.
\end{description}

\begin{figure}
\centering
\includegraphics[width=\linewidth]{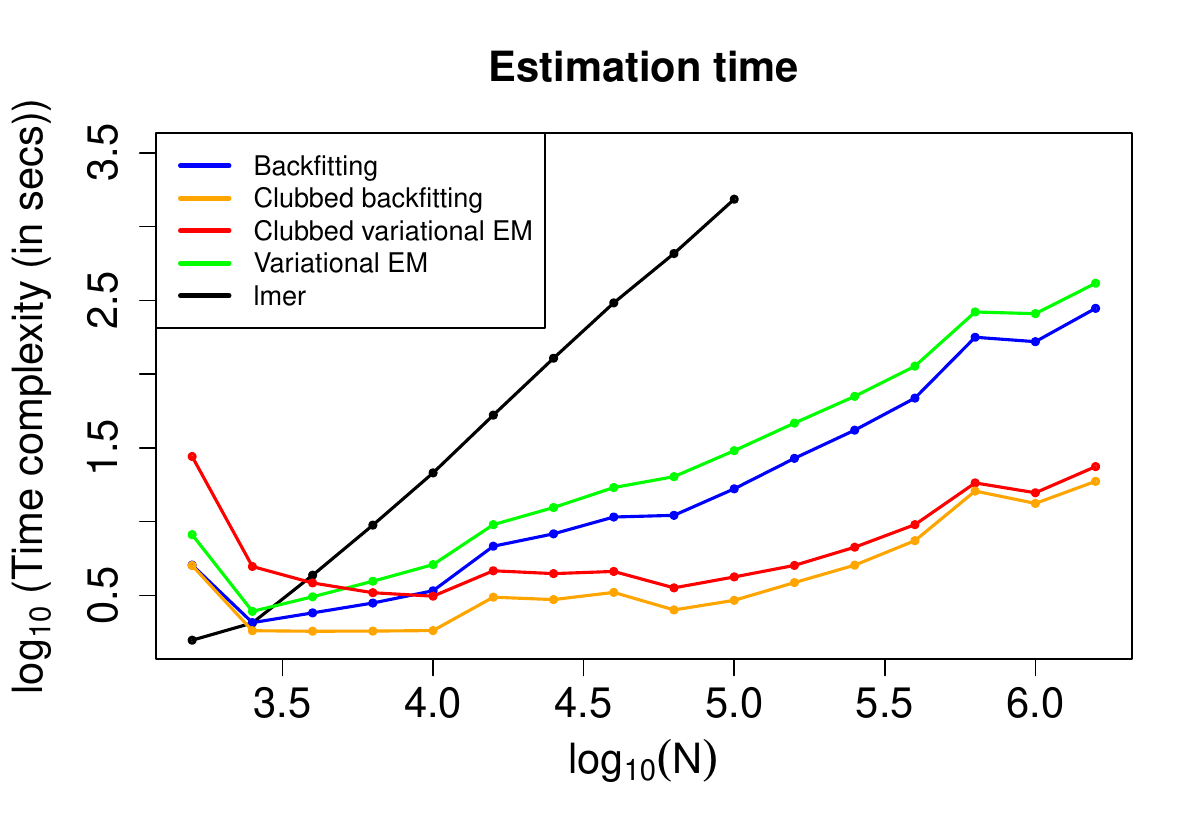}
\caption{Diagonal case: Comparison of estimation time and number of iterations by backfitting, clubbed backfitting, variational EM, and clubbed variational EM (four proposed algorithms) and ``lmer"(existing algorithm). The plot for lmer is discontinued after $N = 10^{5}$ due to the large computational time for ``lmer".}
\label{fig:diag_estimation_time}
\end{figure}

{\textbf{Summary of estimation time: }The slope of a curve in Figure~\ref{fig:diag_estimation_time} represents the exponent of $N$ in time complexity of the associated algorithm. Therefore, the time complexity of ``lmer" is $\mathcal{O}(N^{3/2})$ or worse. On the other hand, the time complexity of all of our suggested algorithms is linear in $N$. The reason that the algorithms with the ``clubbing" trick seem to have time complexity with much smaller exponent in $N$ is that the number of iterations required for those algorithms to converge decreases with $N$. Therefore, although time complexity of each step is linear in $N$, the total time complexity seems to be smaller than $\mathcal{O}(N)$. For $N = 10^5 = 100, 000$, it takes around three days (on the Stanford Sherlock cluster even with parallel processing using ``mclapply" function in R) to perform 100 simulations using ``lmer" as the fitting  function. Therefore, in the interest of time, we discontinue using ``lmer" for values of $N$ larger than 100, 000.\\

Figure~\ref{fig:diag_sigma_ab} depicts empirical consistency of $\wh{\Sigma}_a$ and $\wh{\Sigma}_b$ estimated through MoM and variational EM as errors explained in $\ref{error_sigma_a_b}$ go to zero as $N$ increases. Figure~$\ref{fig:diag_beta_and_sigma2e}$ shows that the accuracy of different estimates of $\bbeta$, i.e., $\hat{\bbeta}_{\mathrm{GLS}}$ obtained though different approaches and $\hat{\bbeta}_{\mathrm{MLE}}$ coincide, i.e., errors are very close to each other and thus $\hat{\bbeta}_{\mathrm{GLS}}$ is consistent. Figure~$\ref{fig:diag_beta_and_sigma2e}$ also depicts the empirical consistency of the estimate of $\sigma^2_e$ obtained through the method of moments approach and variational EM approach. We see that estimates obtained through MoM approach are consistent but variational EM helps in improving those estimates. The errors made in estimating all the parameters get close to maximum likelihood estimates as $N$ increases. 
\begin{figure}
\centering
\includegraphics[width=\linewidth]{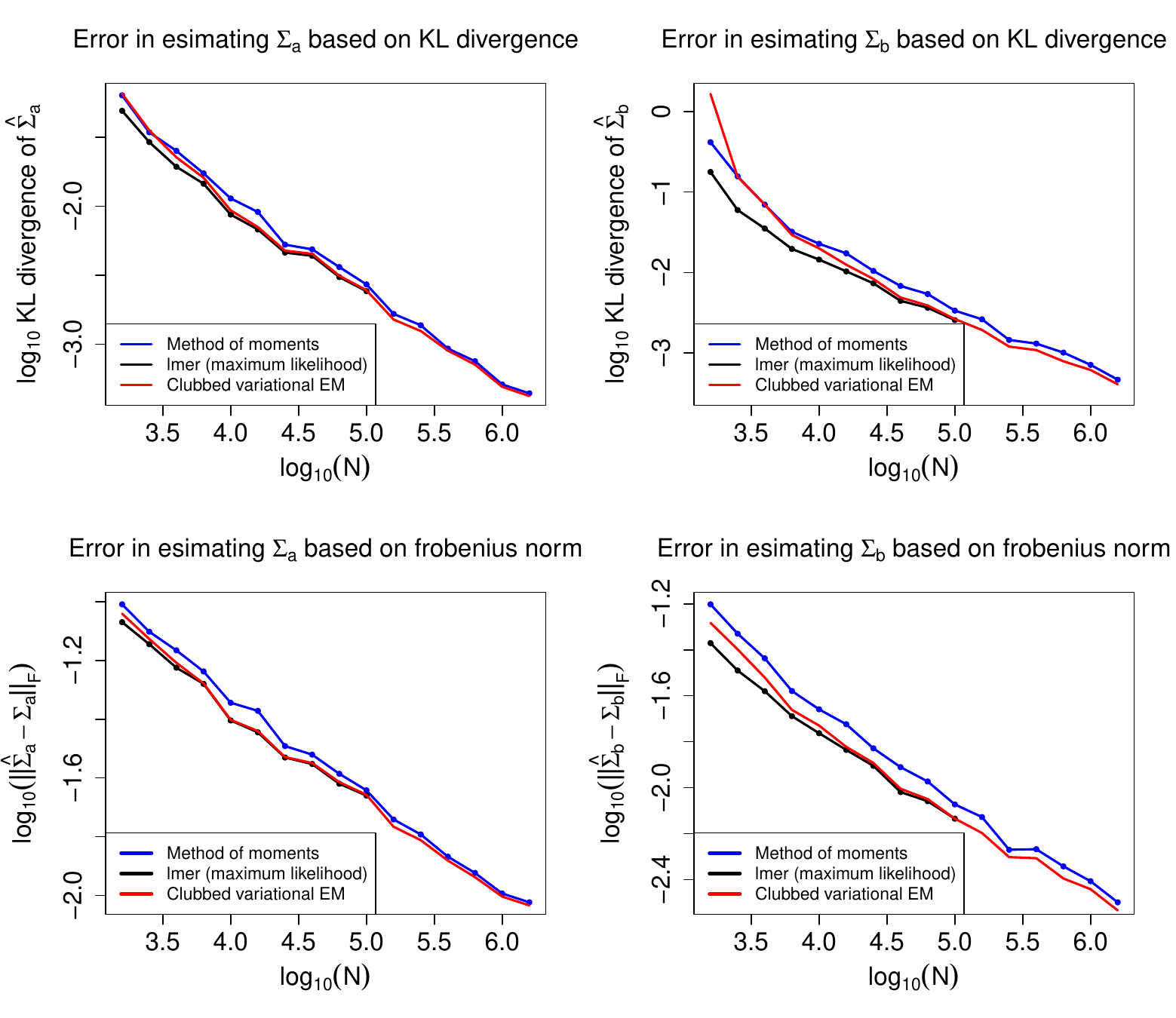}
\caption{Diagonal case: Comparison of accuracy in estimating ${\Sigma}_a$ and ${\Sigma}_b$ based on KL divergence and Frobenius norm by the method of moments (proposed algorithm) and lmer(existing algorithm). The plot for lmer is discontinued after $N = 10^{5}$ due to the large computational time for ``lmer".}
\label{fig:diag_sigma_ab}
\end{figure}
\begin{figure}
\includegraphics[width=\linewidth]{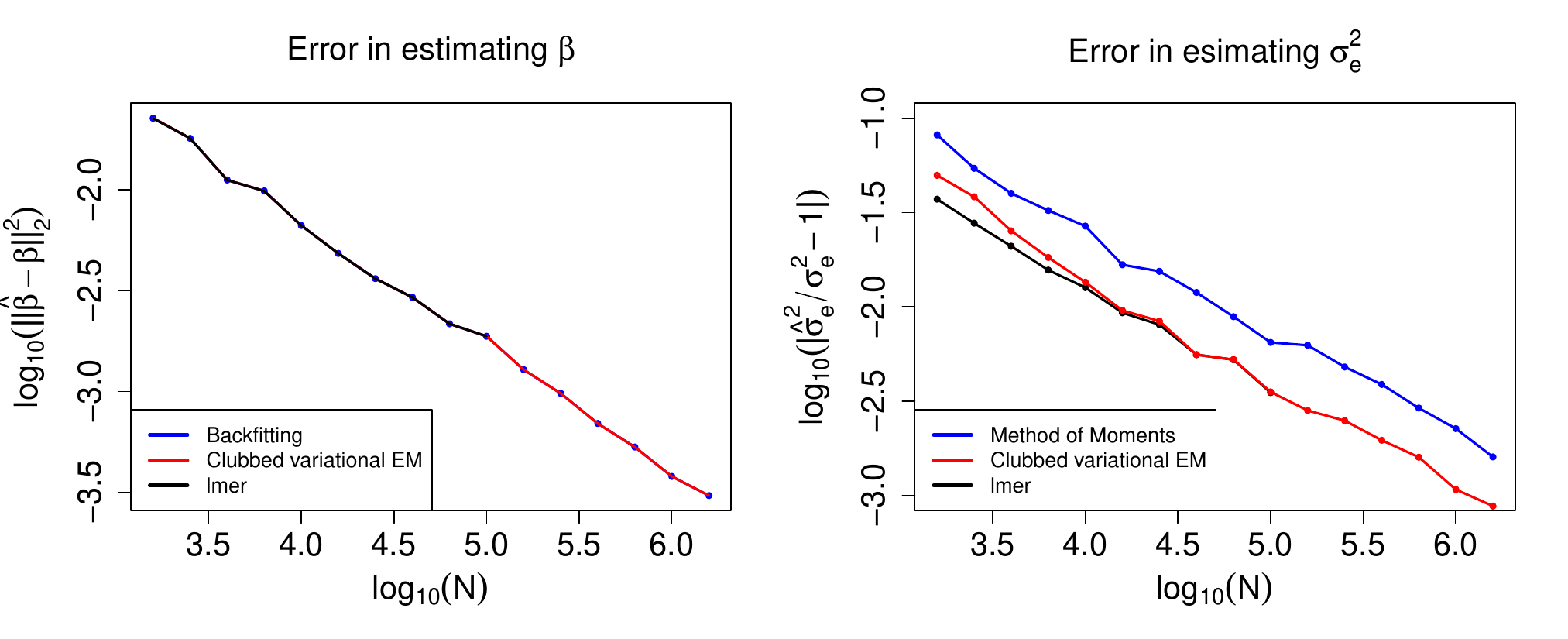}
\caption{Diagonal case: Comparison of accuracy in estimating $\bbeta$ and $\sigma^2_e$ by method of moments (proposed algorithm) and lmer (existing algorithm). The plot for lmer is discontinued after $N = 10^{5}$ due to the large computational time for ``lmer". The errors in estimating $\bbeta$ are close to each other for the illustrated algorithms, which leads the lines in the plot to coincide.}
\label{fig:diag_beta_and_sigma2e}
\end{figure}
\subsubsection{Non diagonal case}
For the non-diagonal case, we also choose $p = 3, \bbeta = (0.2, 0.3, 0.4), \beta_0 = 0.1, \sigma^2_e = 1$. We considered \[\Sigma_a = \Sigma_b = \begin{bmatrix}
1 & 0.2 & 0.2 & 0.2\\
0.2 & 1 & 0.2 & 0.2\\
0.2 & 0.2 & 1 & 0.2\\
0.2 & 0.2 & 0.2 & 1 \\
\end{bmatrix} . \]
We choose $R = C = N^{0.6}$ for our simulations and simulated 100 examples similar to the diagonal case.
\begin{figure}
\centering
\includegraphics[width=\linewidth]{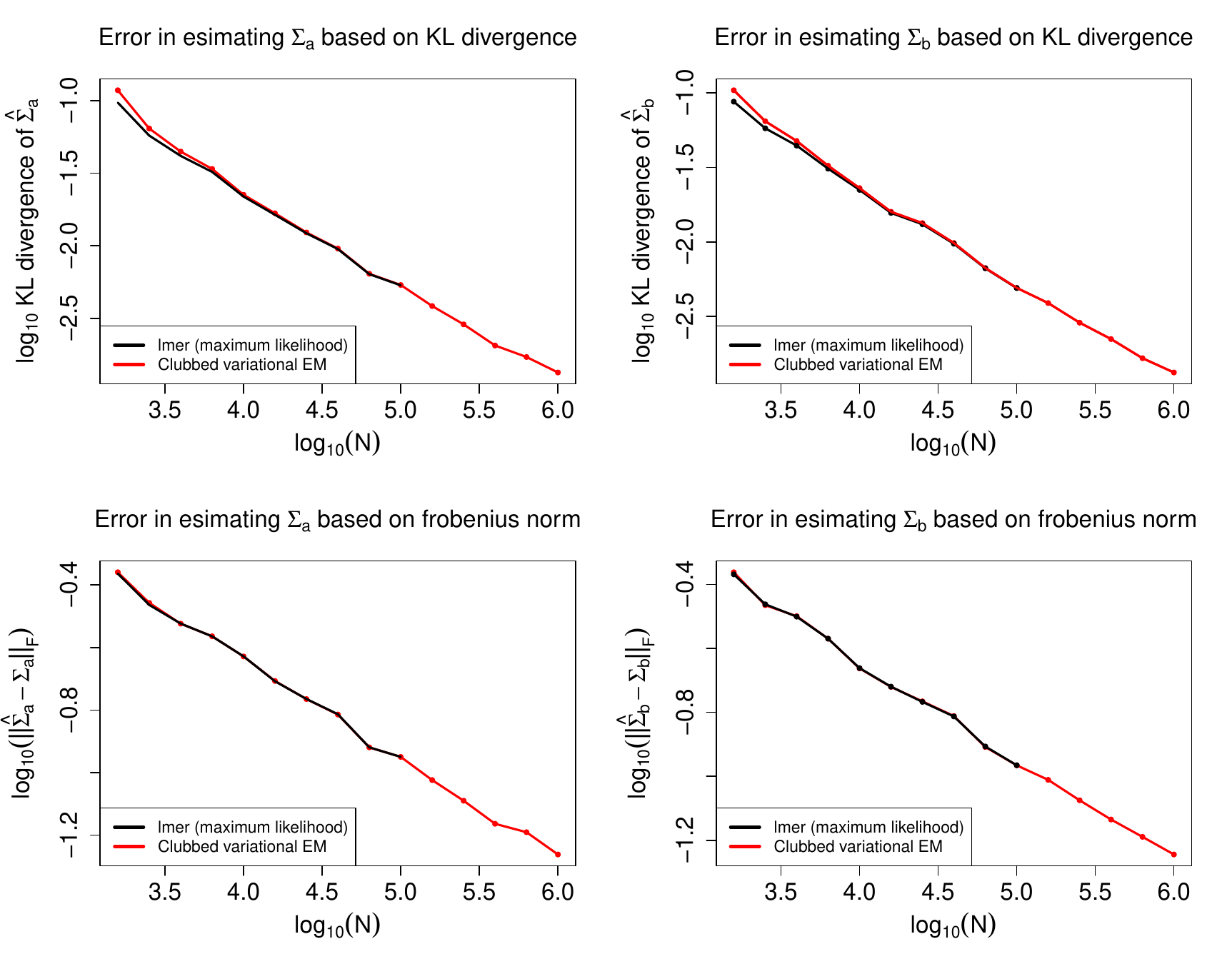}
\caption{Non-diagonal case: Comparison of accuracy in estimating ${\Sigma}_a$ and ${\Sigma}_b$ based on KL divergence and Frobenius norm by variational EM (proposed algorithm) and lmer(existing algorithm). The plot for lmer is discontinued after $N = 10^{5}$ due to large computational time for ``lmer".}
\label{fig:sigma_ab}
\end{figure}
Figure~\ref{fig:sigma_ab} depicts empirical consistency of $\wh{\Sigma}_a$ and $\wh{\Sigma}_b$ estimated through variational EM as errors explained in $\ref{error_sigma_a_b}$ go to zero as $N$ increases. Figure~$\ref{fig:beta_and_sigma2e}$ shows that the accuracy of different estimates of $\bbeta$, i.e., $\hat{\bbeta}_{\mathrm{GLS}}$ obtained though different approaches and $\hat{\bbeta}_{\mathrm{MLE}}$ coincide and thus $\hat{\bbeta}_{\mathrm{GLS}}$ is consistent. The estimates of $\sigma^2_e$ obtained through MoM approach and variational EM depict empirical consistency in figure~$\ref{fig:beta_and_sigma2e}$ although the estimate obtained through variational EM has a faster rate of convergence. We don't plot the error in estimating covariance matrices by MoM approach as method of moments is not fruitful in estimating covariances matrices of random slopes when they aren't diagonal. Figure~$\ref{fig:estimation_time}$ shows that our variational EM and backfitting approach take a great lead compared to the ``lmer" approach. The average time complexity of ``lmer" approach is at least $\mathcal{O}(N^{1.5})$ in our setup, while all the approach discussed by us have time complexity $\mathcal{O}(N)$. The number of iterations required to converge grows with $N$ for ``lmer" and the approaches without clubbing, while the number of iterations decreases with $N$ for clubbing-based approaches. Here is a summary of the results:
\begin{table}[ht]
\centering
\begin{tabular}{|l|l|l|}
\hline
& Without Clubbing & With Clubbing\\
\hline
Backfitting & 
Slow convergence & 
Fast convergence\\
with MoM & 
Poor estimates for $\Sigma_a, \Sigma_b$ & 
Poor estimates for $\Sigma_a, \Sigma_b$\\
\hline
Variational EM & 
Slow convergence & 
Fast convergence\\
& 
Good estimates for $\Sigma_a, \Sigma_b$ & 
Good estimates for $\Sigma_a, \Sigma_b$\\
\hline
\end{tabular}
\end{table}
\begin{figure}
\centering
\includegraphics[width=\linewidth]{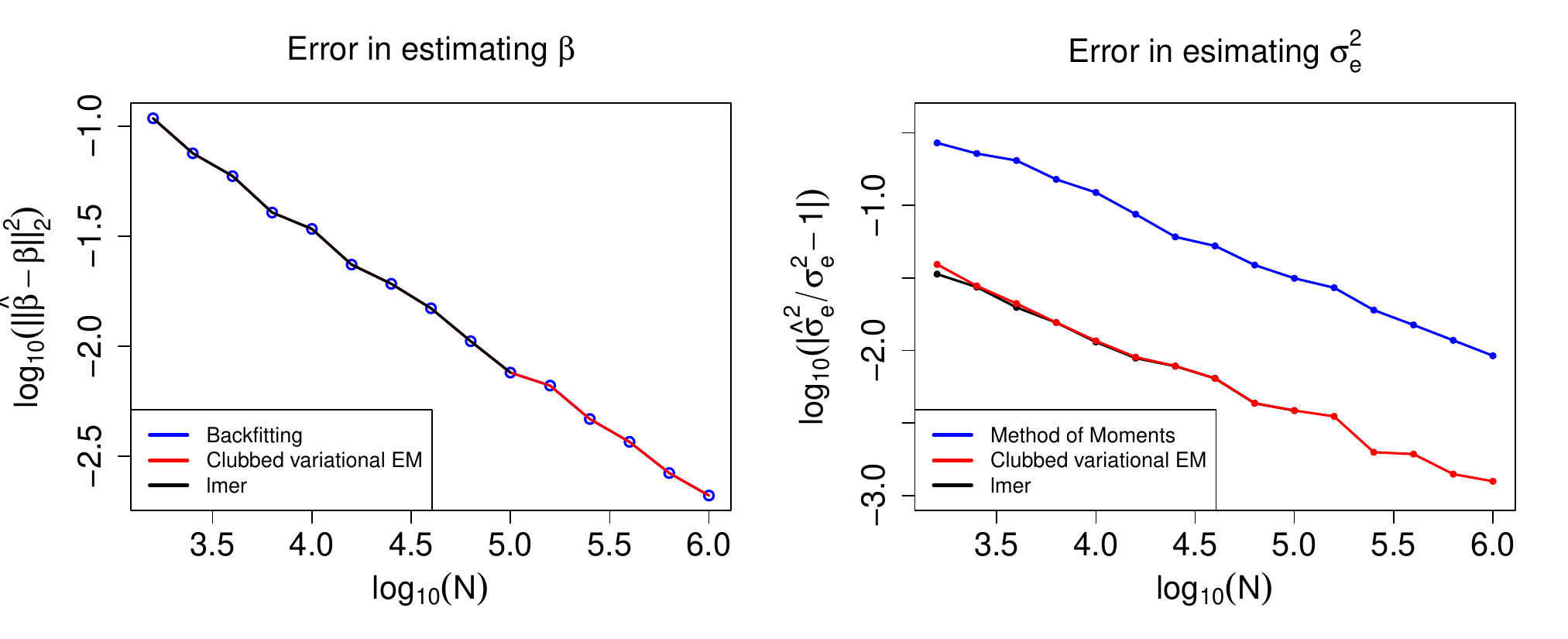}
\caption{Non diagonal case: Comparison of accuracy in estimating $\bbeta$ and $\sigma^2_e$ by variational EM (proposed algorithm) and lmer(existing algorithm). The plot for lmer is discontinued after $N = 10^{5}$ due to the large computational time for ``lmer". The errors in estimating $\bbeta$ are close to each other for the illustrated algorithms, which leads the lines in the plot to coincide.}
\label{fig:beta_and_sigma2e}
\end{figure}
\begin{figure}[ht]
\centering
\includegraphics[scale=0.55]{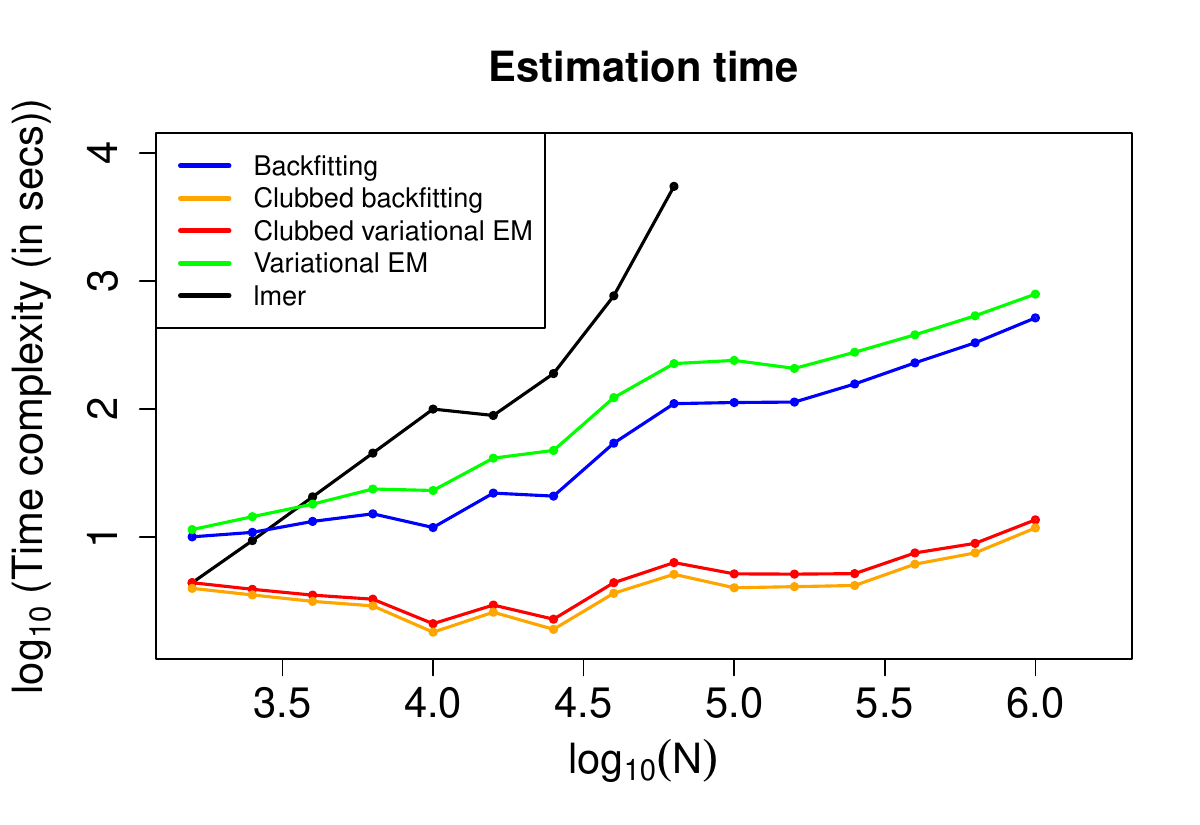}
\caption{Non-diagonal case: Comparison of estimation time by backfitting, clubbed backfitting, variational EM, and clubbed variational EM (four proposed algorithms) and ``lmer"(existing algorithm).}
\label{fig:estimation_time}
\end{figure}

\begin{figure}[!ht]
\centering
\includegraphics[scale=0.6]{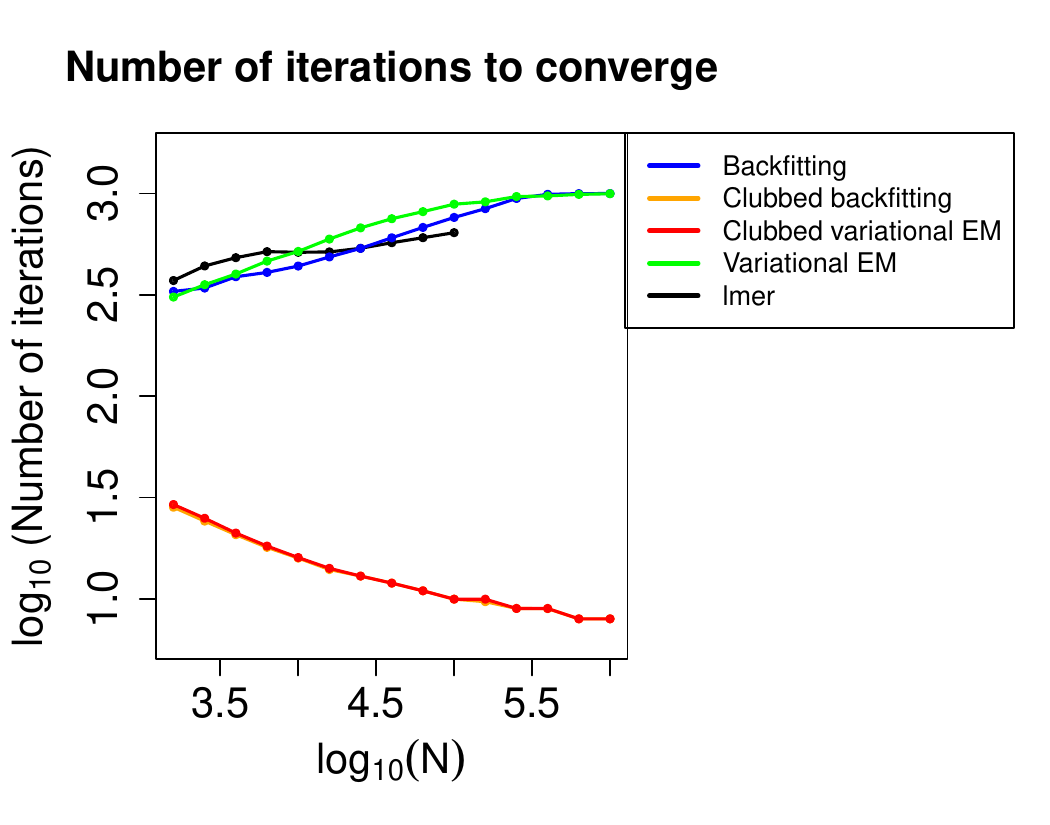}
\caption{Non-diagonal case: Comparison of number of iterations by backfitting, clubbed backfitting, variational EM, and clubbed variational EM (four proposed algorithms) and ``lmer"(existing algorithm). The plot for lmer is discontinued after $N = 10^{5}$ due to the large computational time for ``lmer".}
\label{fig:niter}
\end{figure}
Thus, clubbed variational EM provides precise estimates of all the covariance parameters with a small estimation time.
\subsection{Application of crossed random effects model on `Stitch Fix" data}\label{sec:appl-cross-rand}

~\cite{ghos:hast:owen:2021} applied the backfitting algorithm to obtain $\hat{\bbeta}_{\gls}$ for the Stitch Fix data for the random intercepts model $(\ref{eq:intercept_model})$. Here, we demonstrate how our algorithm can be used to apply the crossed-random effects models to the Stitch Fix data using two models. The first model we implement considers a random intercept and a random slope for a variable named ``match". ``Match" is a continuous variable taking values in the range 0 to 1, representing the predicted probability that the client will keep the item purchased. The second model we implement considers random intercept and random slopes for the variable ``match" and multiple indicator variables named ``client edgy", ``client boho", ``item edgy", ``item boho", etc. These refer to fashion styles that describe some items and some clients' preferences. Some customers and some items were given the adjective ‘edgy’ in the data set and another adjective was ‘boho’, short for ‘Bohemian’. In both of the models, we assume multivariate normal distributions for the random effects for clients as well as items. It should be noted that the purpose of this section is not to choose the most appropriate random effects model. Instead, we wish to illustrate the application of our algorithm in a setup with random slopes. Along with fitting the models, we also measured the prediction accuracy of these models and compared them with ordinary least squares and the model with just random intercepts.
\subsubsection{\textbf{Model 1: }random slope for the variable ``match"}
\label{sec.stitch_fix_model_1}
Here, we implement our devised algorithm to fit the random effects model with a random slope for the ``match" variable. Thus, the model we consider here is
\begin{equation}
\label{eq:Model_1}
y_{ij} =\beta_0 + \tilde{\mathbf{x}}^{\tran}_{ij}(\bsa_i +\bsb_j) + \mathbf{x}_{ij} ^{\tran}\bbeta + \varepsilon_{ij} \hspace{0.5cm} \forall i \in \{1,\cdots, R\} \hspace{0.3cm} j\in \{1,,\cdots,C\},
\end{equation}
where, $\bbeta$ denotes the fixed effect for the entire covariate vector $\mathbf{x}_{ij} $ listed in table~\ref{Tab:stitch_fix}. We assume that ${\bsa}_i \sim \mathcal{N}_2(0, {\Sigma}_a)$ and ${\bsb}_j \sim \mathcal{N}_2(0, {\Sigma}_b)$ represent random effects for $\tilde{\mathbf{x}}_{ij} = (1, \mathrm{match}_{ij})$. The estimate of fixed effect is tabulated in table~$\ref{Tab:stitch_fix}$ as $\hat{\bbeta}_{\gls,1}$ and the fixed effects for random intercepts model is tabulated as $\hat{\bbeta}_{\gls,0}$. The estimate of $\sigma^2_e$ is 4.454 and that of $\Sigma_a$ and $\Sigma_b$ is
$$\wh{\Sigma_a} = \begin{pmatrix} 1.885 & -1.197\\
-1.197 & 1.612\\
\end{pmatrix}
\hbox{ and } \wh{\Sigma_b} = \begin{pmatrix}
0.831 & -1.623\\
-1.623 & 4.297
\end{pmatrix} .$$

\subsubsection{\textbf{Model 2: }random slope for variables ``match", ``I(item edgy)",  and ``I(item boho)" for clients and random slope for variables ``match", ``I(client edgy)", and  ``I(client boho)"  for items}
\label{sec.stitch_fix_model_2}
The model considered here is
\begin{equation}
\label{eq:Model_2}
y_{ij} = \beta_0+\mathbf{x}^\tran_{a, ij}\bsa_i +\mathbf{x}^\tran_{b, ij} \bsb_j+ \mathbf{x}_{ij} ^{\tran}\bbeta + \varepsilon_{ij} \hspace{0.5cm} \forall i \in \{1,\cdots, R\} \hspace{0.3cm} j\in \{1,,\cdots,C\}, \,
\end{equation}
where, $\bbeta$ denotes the fixed effect for the entire covariate vector $\mathbf{x}_{ij}$ similar as above. We assume that ${\bsa}_i \sim \mathcal{N}_4(0, {\Sigma}_a)$ represent random effects for
\begin{align*}
\mathbf{x}_{a, ij} & = (1, \mathrm{match}_{ij},\text{I(item edgy})_j, \text{I(item boho)}_j)
\end{align*}
and ${\bsb}_j \sim \mathcal{N}_4(0, {\Sigma}_b)$ represent random effects for
\begin{align*}
\mathbf{x}_{b, ij} & = (1, \mathrm{match}_{ij},\text{I(client edgy})_i,\text{I(client boho)}_i)
\end{align*}
The estimate of fixed effect is tabulated in table~$\ref{Tab:stitch_fix}$ as $\hat{\bbeta}_{\gls,2}$  and the estimate of $\sigma^2_e$ is 4.363 which is smaller than estimated variance of noise estimated for Model 1. The estimate of $\Sigma_a$ and $\Sigma_b$ is
$$
\wh{\Sigma_a} =
\begin{pmatrix}
1.751 & -1.015 & -0.009 & 0.048 \\ 
   -1.015 & 1.422 & -0.014 & -0.034 \\ 
   -0.009 & -0.014 & 0.340 & 0.047\\ 
   0.048 & -0.034 & 0.047 & 0.490
\end{pmatrix}
$$
$$
\wh{\Sigma_b} = \begin{pmatrix}
0.834 & -1.574 & 0.000 & -0.022 \\ 
 -1.574 & 4.185 & -0.032 & -0.006 \\ 
 0.000 & -0.032 & 0.024 & -0.001 \\ 
 -0.022 & -0.006 & -0.001 & 0.022 \\ 
\end{pmatrix}. $$

\subsection{Naivety and inefficiencies with respect to OLS}
~\cite{ghos:hast:owen:2021} defined naivety and inefficiency of $\ols$ with respect to the random intercepts model to establish the utility of fitting a random effects model. Ordinary least squares usually make two mistakes on crossed random effects data. The first is that OLS is naive about correlations in the data and this can lead it to severely underestimate the variance of $\hat{\bbeta}$. The second is that OLS is inefficient compared to GLS by the Gauss-Markov theorem. Let $\hat{\bbeta}_{\ols}$ and $\hat{\bbeta}_{\gls}$ be the OLS and GLS estimates of $\bbeta$, respectively. They quantified the naivety and inefficiency of OLS as follows:
\begin{itemize}
\item \textbf{Naivety} is defined by the ratios $\var_{\gls}(\hat{\bbeta}_{\ols,j})/\var_{\ols}(\hat{\bbeta}_{\ols,j})$ for $j \in \{1, \cdots, p \}$
\item \textbf{Inefficiency} is defined by the ratios  $\var_{\gls}(\hat{\bbeta}_{\ols,j})/\var_{\gls}(\hat{\bbeta}_{\gls,j})$ for $j \in \{1, \cdots, p \}$.
\end{itemize}

We carry the work forward by displaying the naivety and inefficiency of $\ols$ with respect to the random effects model discussed above. For the first model, the values of naivety range from 2.63 to 933.47 while for the second model, the range is 3.80 to 915.86. These values are larger than naivety values reported in ~\cite{ghos:hast:owen:2021} as the largest value reported by them was 345.28. The largest and second-largest ratios are for material indicators corresponding to ‘Modal’ and ‘Tencel’ using the models with random slopes too. We also identified the linear combination of $\hat{\bbeta}_{\ols}$ for which OLS is most naive. We maximize the ratio $ \mathbf{x}^\tran \wh\cov_{\gls,1}(\hat\bbeta_{\ols})x/\mathbf{x}^\tran \wh\cov_{\ols}(\hat{\bbeta}_{\ols})x$ over $x \ne 0$. The resulting maximal ratio is the largest eigenvalue of $\wh\cov_{\ols}(\hat\bbeta_{\ols})^{-1}\wh\cov_{\gls,1}(\hat\bbeta_{\ols})$ and it is about 1008 while the ratio provided for the random intercepts model in ~\cite{ghos:hast:owen:2021} was 361. The resulting maximal ratio for $\hat \bbeta_{GLS, 2}$ is about 988. We also quantified inefficiency similar to~\cite{ghos:hast:owen:2021}. The figure~$\ref{fig:naivety_inefficiency}$ and $\ref{fig:naivety_inefficiency_mc}$ plots the naivety and inefficiency values. The efficiency values range from 1.06 to 32.60 for the first model and can be interpreted as factors by which using random slopes for the variable ``match" reduces the effective sample size. The efficiency values for the second model range from 1.07 to 32.43, therefore adding the random slope for other variables didn't decrease the effective sample size. The ``match" variable is found to be an outlier in terms of efficiency for both of the random slopes models similar to the model with just random intercepts. An important thing to note here is that though OLS is much more naive with respect to the models with random slopes compared to the model with just random intercepts, although the difference in efficiency is not that huge.

\begin{figure}[!ht]
\centering
\includegraphics[width=\linewidth]{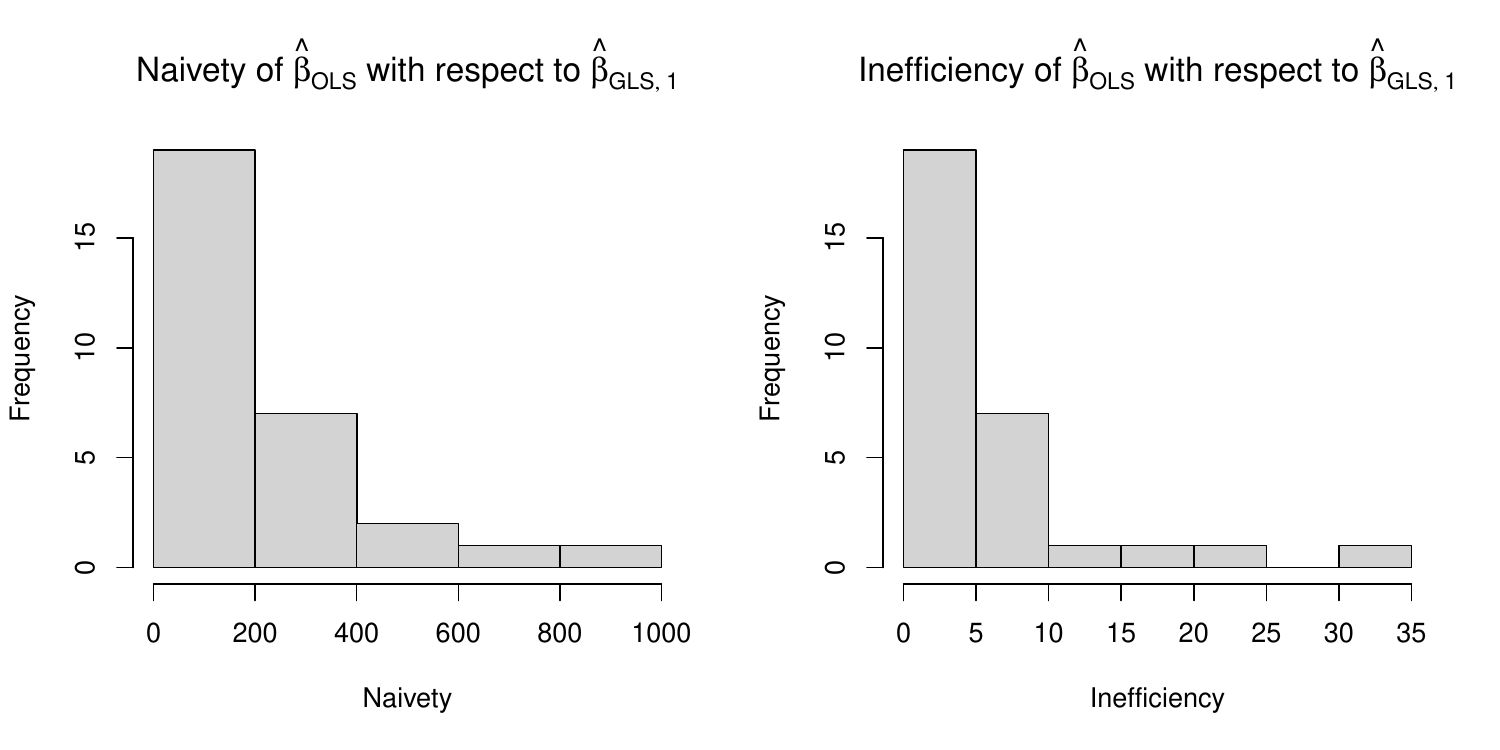}
\includegraphics[width=\linewidth]{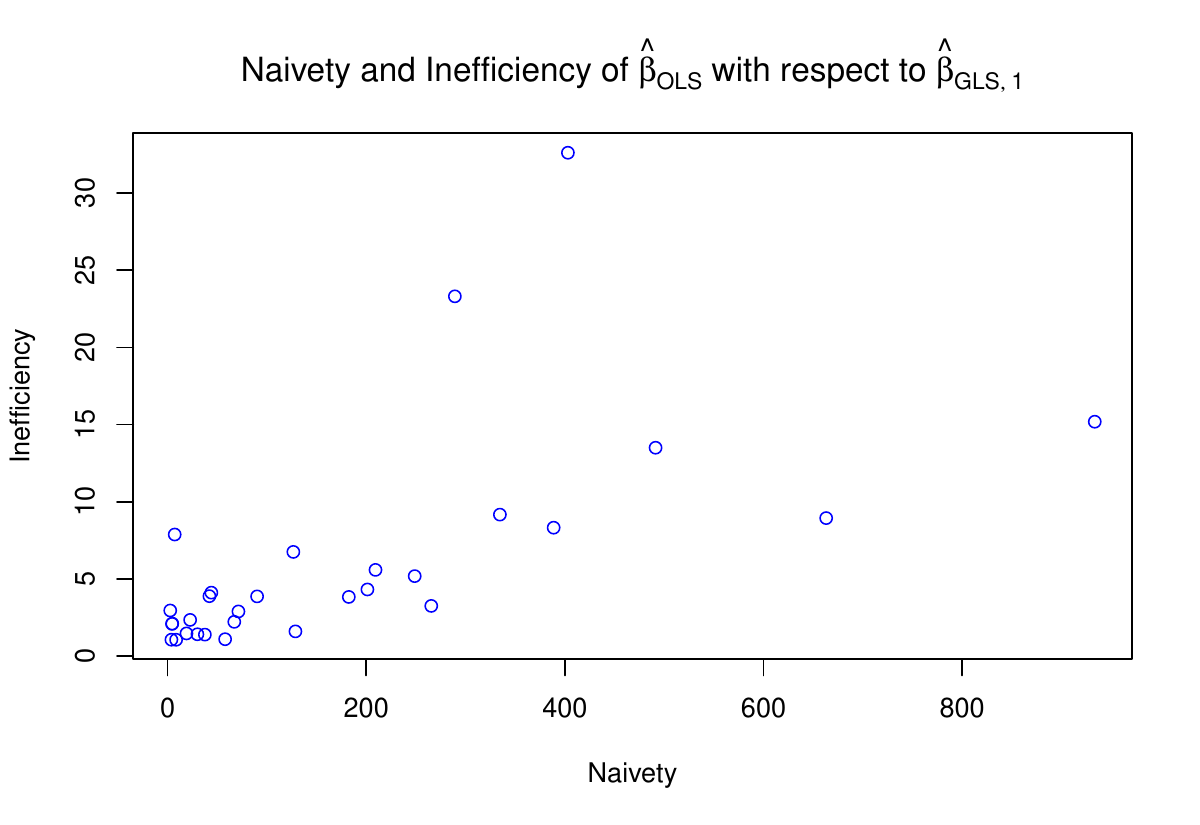}
\caption{Naivety and efficiency of $\hat\bbeta_{\ols}$ with respect to $\hat\bbeta_{\gls,1}$}
\label{fig:naivety_inefficiency}
\end{figure}

\begin{figure}[!ht]
\centering
\includegraphics[width=\linewidth]{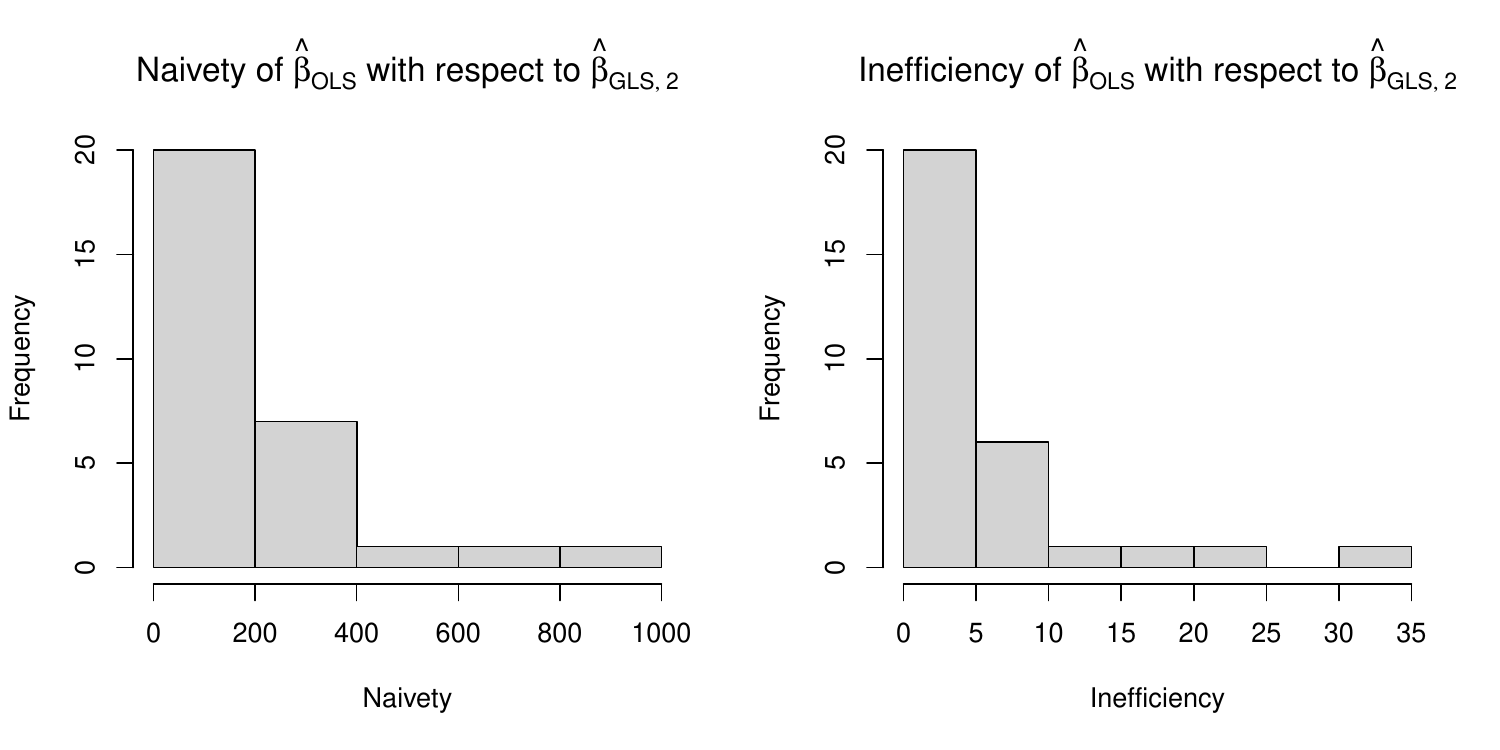}
\includegraphics[width=\linewidth]{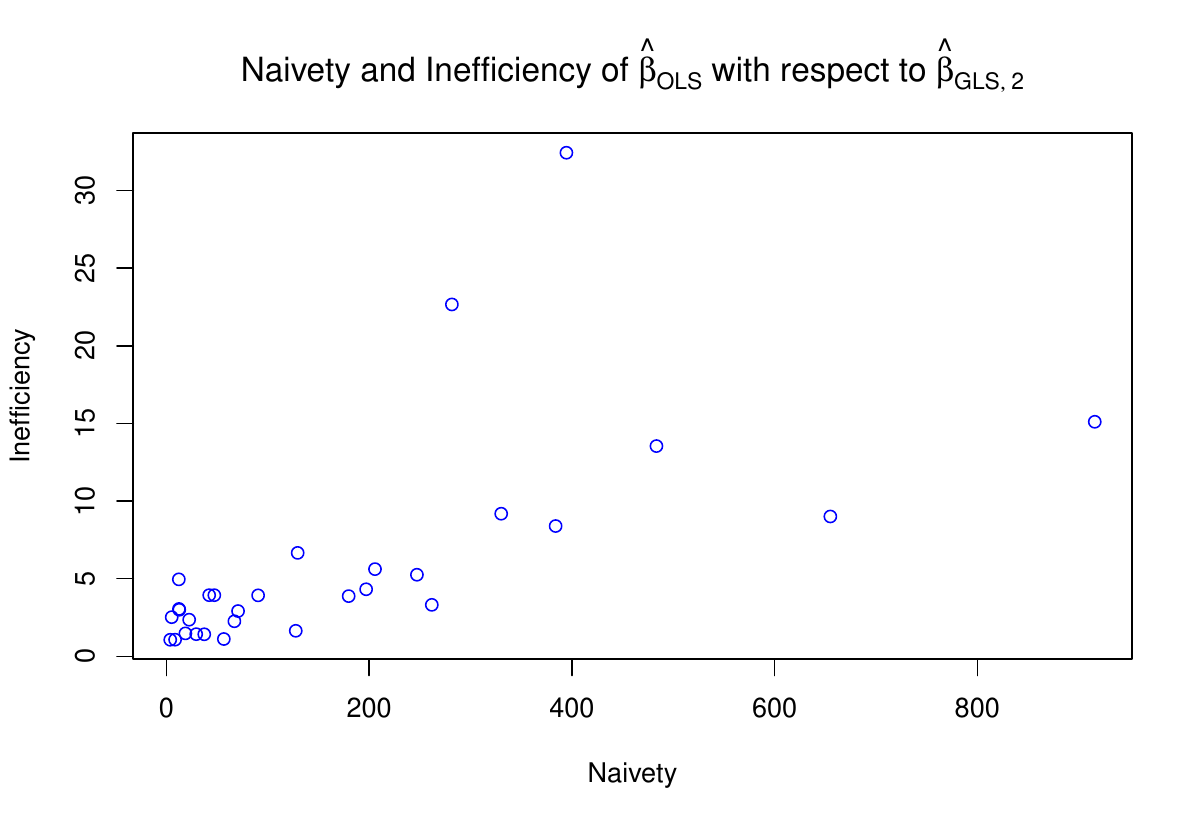}
\caption{Naivety and efficiency of $\hat\bbeta_{\ols}$ with respect to $\hat\bbeta_{\gls,2}$}
\label{fig:naivety_inefficiency_mc}
\end{figure}

\subsection{Fixed effects for Stitch Fix data}
\begin{table}[!ht]
\begin{center}
\begin{tabular}{|l|l|l|l|l|}
\hline
& $\hat\bbeta_{\ols}$ & $\hat\bbeta_{\gls.0}$ & $\hat\bbeta_{\gls.1}$ & $\hat\bbeta_{\gls.2}$  \\ \hline
Match & 5.048*** & 3.442*** & 3.296*** & 3.290***\\ \hline
$\mathbb{I}\{\text { client edgy }\}$ & 0.00102 & 0.00304 & -0.0004 & 0.0005 \\ \hline
$\mathbb{I}\{\text { item edgy }\}$& -0.3358*** & -0.3515*** & -0.3667*** & -0.3885***\\ \hline
\shortstack{$\mathbb{I}$\{client edgy \}\\ $\times \mathbb{I}$\{ item edgy \}} &0.3925***& 0.3793***& 0.3772***& 0.4063***\\ \hline
$\mathbb{I}\{\text { client boho }\}$ & 0.1386*** & 0.1296*** & 0.1312*** & 0.1396***\\ \hline
$\mathbb{I}\{\text { item boho }\}$ & -0.5499*** & -0.6266*** & -0.6270*** & -0.6599***\\ \hline
\shortstack{$\mathbb{I}$\{client boho \}\\ $\times \mathbb{I}$\{ item boho \}} &0.3822*** &0.3763*** & 0.3731*** & 0.4049***\\ \hline
Acrylic & -0.06482*** & -0.00536 & 0.0122 & 0.0468 \\ \hline
Angora & -0.0126 & 0.07486 & 0.0978 &  0.1178\\ \hline
Bamboo & -0.04593 & 0.03251 & 0.0730 & 0.0839\\ \hline
Cashmere & -0.1955*** & 0.00893 & 0.0261 & 0.0576 \\ \hline
Cotton & 0.1752*** & 0.1033*** & 0.1196*** & 0.1308*** \\ \hline
Cupro & 0.5979* & 0.2089 & 0.0350 & 0.0318\\ \hline
Faux Fur & 0.2759*** & 0.2749*** & 0.2814** & 0.2779** \\ \hline
Fur & -0.2021*** & -0.07924 & -0.0776 & -0.0533\\ \hline
Leather & 0.2677*** & 0.1674* & 0.1531 & 0.1557\\ \hline
Linen & -0.3844*** & -0.08658 & -0.0833 & -0.1089 \\ \hline
Modal & 0.0026 & 0.1388* & 0.1318 & 0.1397\\ \hline
Nylon & 0.0335* & 0.08174 & 0.0899 & 0.0933 \\ \hline
Patent Leather & -0.2359 & -0.3764 & -0.5266 & -0.5448\\ \hline
Pleather & 0.4163*** & 0.3292*** & 0.3510*** & 0.3315*** \\ \hline
PU & 0.416*** & 0.4579*** & 0.4509*** & 0.4618***\\ \hline
PVC & 0.6574*** & 0.9688** & 0.8546 & 0.8812\\ \hline
Rayon & -0.01109*** & 0.05155*** & 0.0654*** & 0.0692*** \\ \hline
Silk & -0.1422*** & -0.1828*** & -0.2186*** & -0.2117*** \\ \hline
Spandex & -0.3916*** & 0.414*** & 0.4080** & 0.4290** \\ \hline
Tencel & 0.4966*** & 0.1234* & 0.1163 & 0.1442\\ \hline
Viscose & 0.04066*** & -0.02259 & -0.0336 & -0.0370\\ \hline
Wool & -0.0602*** & -0.05883 & -0.0689 & -0.0521 \\ \hline 
\end{tabular}
\end{center}
\caption{Each column represents the estimated fixed effects for the
  \emph{Stitch Fix} data under different random effect regimes: OLS
  (no random effects), crossed random effects model with random
  intercepts,  and two instances of the crossed random effects model with random intercepts as well as random slopes. ``Match" is a continuous variable taking values in the range 0 to 1 representing the prediction of the probability that the client will keep the item purchased (based on an independent baseline model). The indicator variables named ``client edgy", ``client boho", ``item edgy", ``item boho" refer to fashion styles that describe some items and some client's preferences. The asterisks used to indicate significance follow the standard R convention, where * represents significance at the 5\% level ($p \le$ 0.05), ** at the 1\% level ($p \le$ 0.01), and *** at the 0.1\% level ($p \le$ 0.001).}
  
\label{Tab:stitch_fix}
\end{table}
\footnotetext{Some customers and some items were given the adjective ‘edgy’ in the data set and another adjective was ‘boho’, short for ‘Bohemian’. The remaining indicator variables represent the textile used in manufacturing the item. }
\subsubsection{Summary}
In table~\ref{Tab:stitch_fix}, we present the estimated fixed effects while fitting linear model, crossed random effects model with random intercepts~(\ref{eq:intercept_model})  and the two crossed random effects models with random intercepts as well as random slopes discussed in Section \ref{sec.stitch_fix_model_1} and \ref{sec.stitch_fix_model_2}. The estimate for fixed effect of the variable ``match" is 5.048 for OLS and the estimate is 3.442, a much smaller value for crossed random effects model with random intercepts. When we add random slopes associated with this variable for clients as well as items, the fixed effect reduces to 3.296. However, considering the randomness, the effect associated with this variable fluctates in the range 3.296 $\pm$ 1.269 for clients and 3.296 $\pm$ 2.073 for items. The textile indicator variables like ``Acrylic", ``Cashmere", ``Fur", ``Linen", etc.,  are found to be significant for the OLS model but they lose their significance on fitting crossed random effects models. A possible reason could be that OLS underestimates the standard errors. There are certain textile variables like ``Leather", ``PVC", and ``Tencel" which are significant while fitting crossed random effects model with random intercepts but they loose significant on adding randomness for slope of 	``match". The explaination once again is that the crossed random effects model with just random intercepts also underestimates the standard errors similar to OLS, which was also implicated in the naivety and efficiency calculations above. Upon getting a more accurate picture of the standard errors, we find fewer covariates to be significant. The direction and magnitude of the fixed effects are mostly consistent between the two crossed random effects models for the variables which are found be significant by both of them.  

\subsection{Prediction error}
We considered a random train-test split of Stitch Fix data. We divide the test data into three categories:
\begin{itemize}
\item \textbf{Repeat client and repeat item:} The category includes test points with clients whose information was seen in the training data and items whose information was also available in the training data. The prediction of the rating was given by
\[\hat{Y}_{ij} = \mathbf{x}_{ij} ^\tran (\hat \bbeta + \hat \bsa_i + \hat \bsb_j)\]
where $\hat \bbeta, \hat {\bsa_i},$ and $\hat {\bsb_j}$ are the estimates of $\bbeta$ and BLUP estimates of $\bsa_i$ and $\bsb_j$ respectively.
\item \textbf{New client and repeat item:} The category includes test points with clients whose information was not seen in the training data and items whose information was available in the training data. The prediction of the rating was given by
\[\hat{Y}_{ij} = \mathbf{x}_{ij} ^\tran (\hat \bbeta + \hat \bsb_j)\]
as the prior expectation for $\bsa_i$ is zero.
\item \textbf{Repeat client and new item:} The category includes test points with clients whose information was seen in the training data and items whose information was not available in the training data. The prediction of the rating was given by
\[\hat{Y}_{ij} = \mathbf{x}_{ij} ^\tran (\hat \bbeta + \hat \bsa_i)\]
as the prior expectation for $\bsb_j$ is zero.
\end{itemize}
We observed only two test data points in the third category, so we are only reporting prediction errors in the first two categories. Similar to table \ref{Tab:stitch_fix}, $\gls_0$ represents the crossed random effects model with random intercepts, $\gls_1$ and $\gls_2$ represent the crossed random effects models with random slopes discussed in Section \ref{sec.stitch_fix_model_1} and  \ref{sec.stitch_fix_model_2} respectively.
\begin{table}[!ht]
\begin{center}
\begin{tabular}
{|c c|c|c|c|}
\hline
& $\ols$ & $\gls_0$ & $\gls_1$ & $\gls_2$ 
\\ 
Repeat client and repeat item & 5.710125 & 4.875282 & 4.864025 & 4.853029
\\
New client and repeat item & 5.682390 & 5.606388 & 5.543627 & 5.550773
\\
\hline
\end{tabular}
\end{center}
\caption{Prediction error for different models among different categories}
\label{table:prediction_errors}
\end{table}
For the first category, we observe 99,115 data points and for the second category, we observe 883 data points. $\gls_0$ provides a significant improvement over $\ols$ in the first category but the improvement is not very significant for the second category as the information for the clients is not available in the training data. $\gls_1$ provides a small improvement over $\gls_0$ in the first category and larger improvement in the second category. $\gls_2$ provides small improvement over $\gls_1$ in the first category, although the performance gets worse in the second category. The overall performance is found best for the second model i.e. $\gls_2$. The results are consistent with the naivety and inefficiency values we discussed before.

\subsection{Application of crossed random effects model on Movie lens data}
In this subsection, we compute the prediction error of crossed random effects model with random slopes on MovieLens dataset and compare it with baseline models like OLS and crossed random effects model with random intercepts. We performed the real data analysis using \emph{MovieLens 100K} dataset available on $\hyperlink{https://www.kaggle.com/datasets/prajitdatta/movielens-100k-dataset}{kaggle}$. The data set describes ratings and free-text tagging activities from MovieLens, a movie recommendation service. It contains around 100,000 ratings from 942 users on 1682 movies. Each user has rated at least 20 movies. The data set also comes with certain information on users and movies, e.g., age, gender of users and genre of the movies. We use the combined information of user and movies as feature vector for linear model, i.e., $$\mathbf{x}_{ij}  = (\text{age of user } i, \text{gender of user } i, \text{one hot encoding for genre of movie } j).$$ Following \cref{remark:randomslopes}, we use subset of client-specific features as covariates with random slope for items and vice-versa. Therefore, we either use just random intercepts for clients, or we use random intercepts supplemented by random slopes for the one hot encoding vector of genre of the movie. Similarly, we either use just random intercepts for items or we use random intercepts supplemented by random slopes for age and gender of the users. We represent these models by $\text{Model}_{00}$, $\text{Model}_{01}$, $\text{Model}_{10}$, and $\text{Model}_{11}$, i.e., $\text{Model}_{ij}$ is a model with random slopes for clients if $i$ = 1 and random slopes for items if $j$ = 1. 

\begin{table}[!ht]
\label{movie_lens_MSE}
\begin{center}
\begin{tabular}{|c|c|c|c|c|c|}
  \hline
 & OLS & $\text{Model}_{00}$ & $\text{Model}_{01}$ & $\text{Model}_{10}$ & $\text{Model}_{11}$ \\ 
  \hline
MSE & 1.224 & 0.885 & 0.878 & 0.881 & 0.875 \\ 
   \hline
\end{tabular}
\end{center}
\caption{Prediction error of ratings in MovieLens data }
\end{table}
$\text{Model}_{00}$ is the crossed model effects model with just random intercepts discussed by \cite{ghos:hast:owen:2021} and it serves as another baseline model for us. The purpose of this section was to illustrate the utility of using crossed random effects model with random slopes. As the models with random slopes helps in bringing significant improvement in prediction accuracy, it establishes the need of scalable solution to crossed random effects model. These models could serve as efficient recommender systems and they come with an additional advantage of being interpretable.
\section{Conclusion}
In the paper, we discussed estimation under crossed random effects model with random slopes. We suggested two different algorithms based on the method of moments combined with backfitting and variational EM. The former is most appropriate when the covariance matrices for the random effects are diagonal and the latter is useful when the covariance matrices are unstructured. The algorithm based on the method of moments was shown to have theoretical as well as empirical consistency to estimate fixed effects and covariance parameters when covariance matrices for random effects were diagonal. The algorithm based on variational EM was able to accomplish a more difficult task which is to estimate the parameters when the covariance matrices are unstructured and showed great empirical results. Both of the algorithms could complete the task to estimate fixed and covariance parameters in $\mathcal{O}(N)$ and are superior to ``lmer" which takes at least $\mathcal{O}(N^{3/2})$. The algorithms find huge applications in settings with large sample sizes and where the time taken to obtain the fit is crucial. So for example, for a problem with $N=100,000$ ``lmer" might take over 40 minutes, while our methods take under 10 seconds; for problems with $N=1,000,000$ our methods take around 30 seconds (see Figure~\ref{fig:diag_estimation_time}.) We demonstrated the performance of our algorithms using simulations as well as real data.

%% file: proofs_appendix_random_slopes.tex
\subsection{Proof of Theorem~\ref{thm:robinson}}\robin* 
\label{Alternative_proof}
Here,
\begin{equation*}
\hat\bbeta_{\gls} = \argmin_{\bbeta} (Y-X\bbeta)\mathcal{V}^{-1}(Y -X\bbeta)=(X^{\tran}\mathcal{V}^{-1}X)^{-1}X^{\tran}\mathcal{V}^{-1}Y \tag{\ref{eq:beta_gls}}
\end{equation*}
\begin{proof}
We rewrite the model as
\[Y=X \bbeta + Z U + \varepsilon \, \, \text{ where, } U \sim \mathcal{N}(0, \Sigma_u), \text{ and } \epsilon \sim \mathcal{N}(0, \sigma^2_e I)\]
Then,
\[P_r(Y = y, U=u) = P_r(Y = y) P_r(U=u| Y =y) = P_r(Y =y|U = u)P_r(U = u) \]
$ P_r(Y = y)$ corresponds to the likelihood in $(\ref{eq:likelihood})$ and $\hat\bbeta_{\gls}$ in $(\ref{eq:beta_gls})$, and $ P_r(Y =y|U = u) P_r(U = u)$ corresponds to the exponent of $(\ref{eq:loss}$). Now,
\begin{align*} P_r & (U=u|Y = y) \\
& = \frac{1}{\sqrt{|Var(U|y)|}} \exp \left(-\frac{(u-\Ex[U|y])^{\tran} \left[Var(U|y)\right]^{-1} (u-\Ex[U|y])}{2} \right) \end{align*}
which is maximized for $u=\Ex[U|Y = y]$, and
$$\max_{u} P_r(U=u|Y = y) = \frac{1}{(2\pi)^{(p+1)/2}\sqrt{|Var(U|y)|}} = \mbox{a constant }$$
\[\max_{u} P_r(y, U=u) = \max_{u} P_r(Y = y) P_r(U=u|Y = y) = \mbox{ constant } \times p(y)\]
Therefore,
\[\argmax_{\bbeta} \max_{u} P_r(Y = y, U=u) = \argmax_{\bbeta} P_r(Y = y) \frac{1}{\sqrt{|Var(U|y)|}}\]
It is understood that $Var(U|y)$ is constant in terms of $\bbeta$, which implies,
\[\argmax_{\bbeta} \max_{u} P_r(Y = y, U=u) = \argmax_{\bbeta} P_r(Y = y) \]
which is equivalent to saying that
$$\hat\bbeta = \hat\bbeta_{\gls}$$
\end{proof}
\subsection{Consistency of \texorpdfstring{$\hat{\bbeta}_{\ols}$}{beta ols}}
\betaols*
\begin{proof}
To begin with, we assume the observations are arranged by the level of the first factor.
$Y = X\bbeta + \varepsilon$ where $\varepsilon \sim N(0,\mathcal{V})$ for
\begin{align*}
\mathcal{V}& = \underbrace{diag(X_{1}^{(a)}\Sigma_{a}X_{1}^{(a) \tran},\cdots,X_{R}^{(a)}\Sigma_{a}X_{R}^{(a) \tran})}_{B_{R}} \\
& \hspace{2cm}+ P_{C} \underbrace{diag(X_{1}^{(b)}\Sigma_{b}X_{1}^{(b) \tran},\cdots,X_{C}^{(b)}\Sigma_{b}X_{C}^{(b)\tran})}_{B_{C}} P_{C}^{\tran} + \sse\bf{I}_{N}
\end{align*}

with $P_{C}$ being an appropriate permutation matrix. Note that $\hat{\bbeta}_{\ols} = \bbeta + (X^{\tran}X)^{-1}X^{\tran}\varepsilon.$ $\Ex((X^{\tran}X)^{-1}X^{\tran}\varepsilon) = \bf{0}.$
Let $w$ be any $p \times 1$ vector with finite norm.
\begin{align*}
\mathrm{Var}(w^{\tran}(X^{\tran}X)^{-1}X^{\tran}\varepsilon) &= w^{\tran}(X^{\tran}X)^{-1}X^{\tran}\mathcal{V}X(X^{\tran}X)^{-1}w\\
&\leq \lambda_{\max}(B_{R}) w^{\tran}(X^{\tran}X)^{-1}w \\
& \hspace{0.2cm}+ \lambda_{\max}(P_CB_{C}P_C^{\tran}) w^{\tran}(X^{\tran}X)^{-1}w + \sse w^{\tran}(X^{\tran}X)^{-1}w \\
&\leq \lambda_{\max}(B_{R}) \frac{\Vert w \Vert^{2}}{cN} + \lambda_{\max}(P_CB_{C}P_C^{\tran}) \frac{\Vert w \Vert^{2}}{cN} + \sse \frac{\Vert w \Vert^{2}}{cN}\\
&\leq \lambda_{\max}(B_{R}) \frac{\Vert w \Vert^{2}}{cN} + \lambda_{\max}(B_{C}) \frac{\Vert w \Vert^{2}}{cN} + \sse \frac{\Vert w \Vert^{2}}{cN} \to 0
\end{align*}
So $\hat{\bbeta}_{\ols}$ is consistent estimator for $\bbeta$.
\end{proof}
\subsection{Method of Moments}\label{appendix.mom}
\mmconsis*
\begin{proof}
Each of the equations (\ref{eq:mm1}), (\ref{eq:mm2}), and (\ref{eq:mm3}) can be represented as
\begin{equation}\Ex \bigg[ r^{\tran} (I - Q) r\bigg] = \hat r^{\tran} (I - Q) \hat r \text{ for some } Q \in \mathbb{R}^{N \times N}.
\label{eq:matrix_form}
\end{equation}
For example, $Q$ for $(\ref{eq:mm1})$ is given by
\begin{equation}
\label{eq:Q}
Q = \begin{bmatrix}
\mathlarger{\frac{{x^{(s)}_{1.}}x^{(s)^{\tran}}_{1.}}{\|{x^{(s)}_{1.}}\|^2}} & 0 & \cdots & 0 \\
0 & \mathlarger{\frac{{x^{(s)}_{2.}}x^{(s)^{\tran}}_{2.}}{\|{x^{(s)}_{2.}}\|^2}} & \cdots & 0 \\
0 & 0 & \ddots & 0\\
0 & 0 & \cdots & \mathlarger{\frac{{x^{(s)}_{R.}}x^{(s)^{\tran}}_{R.}}{\|{x^{(s)}_{R.}}\|^2}}\\
\end{bmatrix}
\end{equation}
It is sufficient to show that $(\Ex[r^{\tran}(I-Q)r|X]- \hat{r}^{\tran}(I-Q) \hat{r})/N$ converges to zero $ \forall Q \in \{P_1, \cdots, P_p\}$, $\{P'_1, \cdots, P'_p\}, J_n $ and  $\forall \, \, \Sigma_a, \Sigma_b, \sigma^2_e $ where $J_n = 1_n 1_n^{\tran} / n$.
\begin{align*}
\frac{\hat{r}^{\tran}(I-Q) \hat{r}}{N} & =\frac{ (\hat{r}-r+r)^{\tran}(I-Q) (\hat{r} - r + r)}{N}\\
& = \frac{r^{\tran}(I-Q)r + (\hat{r}-r)^{\tran}(I-Q) (\hat{r} - r) + 2 r^{\tran} (I-Q) (\hat{r} - r)}{N} \\
& = \frac{r^{\tran}(I-Q)r}{N} + \frac{(\hat{\bbeta}_{\ols}-\bbeta)^{\tran} X^{\tran}(I-Q) X (\hat{\bbeta}_{\ols} - \bbeta)}{N}\\
& \hspace{5cm} - \frac{2 r^{\tran} (I-Q) X (\hat{\bbeta}_{\ols} - \bbeta)}{N}
\end{align*}
As $\mathbf{x}_{ij} $ have finite fourth moment $\forall i \in \{1, \cdots, R\} \, j \in \{1, \cdots, C\} $, $X^{\tran} X /N = O_p(1) $, and as $\hat{\bbeta}_{\ols}$ is consistent, second and third terms converge to zero in probability.
Now,
\begin{align*}
\frac{(\hat{r}^{\tran}(I-Q) \hat{r}-\Ex[r^{\tran}(I-Q)r|X])}{N} & = \frac{r^{\tran}(I-Q)r}{N} - \frac{\Ex[r^{\tran}(I-Q)r|X])}{N} + o_p(1)
\end{align*}
Now, \begin{equation} \frac{r^{\tran}(I-Q)r - \Ex[r^{\tran}(I-Q)r]}{N} \to 0 \end{equation}
and
\begin{equation} \frac{\Ex[r^{\tran}(I-Q)r|X] - \Ex[r^{\tran}(I-Q)r]}{N} \to 0 \end{equation}
by the strong law of large numbers assuming finite fourth moment of $\mathbf{x}_{ij} $ $\forall i \in \{1, \cdots, R\} $ and $ j \in \{1, \cdots, C\} $.
Let, $\btheta = (diag(\Sigma_a),diag(\Sigma_b),\sigma^2_e)$, then method of moment estimate, $\btheta_n $ is solution of equations in the form $A_n \btheta = B_n$, i.e. , $\btheta_n = A_n^{-1} B_n$. The above proof guarantees that $A_n \to A$ and $B_n \to A \btheta, \, \, \forall \btheta$. Therefore, $\btheta_n \to \btheta$.
\end{proof}
\subsection{Proof of~(\ref{eq:simplify})}
We know from (\ref{eq:random_slopes}) that
\[r_{ij} = \mathbf{x}_{ij} ^{\tran} (\bsa_i + \bsb_j) + \varepsilon_{ij}\]
It could also be represented as
\[
r = \sum_{s'=0}^{p} X_A^{(s')} A_{s'} + \sum_{s'=0}^{p} X_B^{(s')} B_{s'} + \varepsilon.
\]
Here, $A_{s'} \in \mathbb{R}^R$ and $B_{s'} \in \mathbb{R}^C$ represent the stacked random coefficients of the $s'$-th coordinate for clients and items, respectively, and $X_A^{(s')}$ and $X_B^{(s')}$ denote the corresponding design matrices. Then, we have
\begin{align*}
&\Ex\bigg[\sum_{i=1}^{R} \sum_{j=1}^{C} z_{ij} r^2_{ij}-\sum_{i=1}^{R}\frac{\left(\sum_{j=1}^{C} z_{ij} r_{ij}x^{(s)}_{ij}\right)^2}{\sum_{j=1}^{C} z_{ij}x^2_{ij(s)}} \bigg{|} X\bigg] = \Ex[r^{\tran}(I-Q)r]\\
& = \Ex\left(\sum_{s'=0}^{p}X_A^{(s')}A_{s'}+\sum_{s'=0}^{p}X_B^{(s')}B_{s'} + \varepsilon\right)^{\tran}\left(I -Q\right)\\
& \hspace{3cm}\cdot\left(\sum_{s'=0}^{p}X_A^{(s')}A_{s'}+\sum_{s'=0}^{p}X_B^{(s')}B_{s'} + \varepsilon\right)\\
& = \sum_{s'=0}^{p} E\bigg[(X_A^{(s')}A_{s'})^{\tran}(I-Q)(X_A^{(s')}A_{s'})\bigg]\\
& \hspace{1cm}+\sum_{s'=0}^{p} E\bigg[(X_B^{(s')}B_s')^{\tran}(I-Q)(X_B^{(s')}B_{s'})\bigg] + \sigma^2_e (n-R)\\
& = \sum_{s'=0}^{p}\tr\bigg[X_A^{\left(s'\right)^{\tran}}\left(I-Q\right)X_A^{\left(s'\right)}\bigg]\Sigma_{a,s's'} +\sum_{s'=0}^{p} \tr\bigg[X_B^{\left(s'\right)^{\tran}}\left(I-Q\right)X_B^{\left(s'\right)}\bigg]\Sigma_{b,s's'}\\
& \hspace{9.5cm} + \sigma^2_e (n-R)\\
&=\sum_{s'=0}^{p} \sum_{i=1}^{R}\left[ \sum_{j=1}^{C} z_{ij} x^{(s')^2}_{ij}-\frac{\left(\sum_{j=1}^{C} z_{ij} x^{(s)}_{ij} x^{(s')}_{ij} \right)^2}{\sum_{j=1}^{C} z_{ij} x^{(s)^2}_{ij}}\right]\Sigma_{a,s's'} \\
& \hspace{1cm}+\sum_{s'=0}^{p}\sum_{i=1}^R\left[ \sum_{j=1}^{C}z_{ij} x^{(s')^2}_{ij} -\frac{\sum_{j=1}^{C}z_{ij} x^{(s)^2}_{ij} x^{(s')^2}_{ij}}{\sum_{j=1}^{C} z_{ij} x^{(s)^2}_{ij}}\right]\Sigma_{b,s's'} + \sigma^2_e (n-R)\\
\end{align*}
The first part of the last equality follows directly from the definition of $Q$, the last part could be shown using the following argument:
\\
Let, \[X_B^{(s')} = \begin{bmatrix}
W_1\\
W_2\\
\vdots\\
W_R
\end{bmatrix}\]
\begin{align*}
\tr\bigg[\left(X_B^{(s')}\right)^{\tran}(I-Q)\left(X_B^{(s')}\right)\bigg] & = \sum_{i=1}^R \tr\bigg[W^{\tran}_i \left(I - \frac{x^{(s)}_{i.}{x^{(s)}_{i.}}^{\tran}}{\|x^{(s)}_{i.}\|^2}\right)W_i\bigg]\\
& = \sum_{i=1}^R \tr(W^{\tran}_i W_i) - \frac{\tr\left(x^{(s)^{\tran}}_{i.} W_i W_i ^{\tran} x^{(s)}_{i.}\right)}{\|x^{(s)}_{i.}\|^2}
\end{align*}
Now,
\[W_i W^{\tran}_i = \begin{bmatrix}
{x^{(s')}_{i1}}^2 & 0 & \cdots & 0\\
0 & {x^{(s')}_{i2}}^2 & \cdots & 0\\
0 & 0 & \ddots & 0\\
0 & 0 & 0 & {x^{(s')}_{in_i}}^2\\
\end{bmatrix}\]
Thus,
\[\tr\bigg[\left(X_B^{(s')}\right)^{\tran}(I-Q)\left(X_B^{(s')}\right)\bigg] = \sum_{i=1}^R\left[ \sum_{j=1}^{C}z_{ij} x^{(s')^2}_{ij} -\frac{\sum_{j=1}^{C}z_{ij} x^{(s)^2}_{ij} x^{(s')^2}_{ij}}{\sum_{j=1}^{C} z_{ij} x^{(s)^2}_{ij}}\right]\]
The additional equation we added was given by
\[\hat{r}^{\tran}\bigg(I - \frac{1}{n}11^{\tran}\bigg) \hat{r} = E\bigg[r^{\tran}\bigg(I - \frac{1}{n}11^{\tran}\bigg)r\bigg]\]
L.H.S. of the above equation will be given by
\[\sum_{i=1}^R \sum_{j=1}^{C} z_{ij}\hat{r}^2_{ij}- \frac{1}{n}\left(\sum_{i=1}^R\sum_{j=1}^{C}z_{ij}\hat{r}_{ij}\right)^2\]
For R.H.S., similar to above, we will need the following expressions:
\begin{align*}\tr\bigg[\left(X_A^{(s)}\right)^{\tran}\bigg(I-\frac{1}{n} & 11^{\tran}\bigg) \left(X_A^{(s')}\right)\bigg] \\
& = \sum_{i=1}^{R}\bigg( \sum_{j=1}^{C}z_{ij} x^{(s)}_{ij}x^{(s')}_{ij}-\frac{\sum_{j=1}^{C} z_{ij}x^{(s)}_{ij}\sum_{j=1}^{C}z_{ij} x^{(s')}_{ij}}{n}\bigg)\end{align*}
and
\begin{align*}\tr\bigg[\left(X_B^{(s)}\right)^{\tran}\bigg(I-  \frac{1} {n}& 11^{\tran}\bigg) \left(X_B^{(s')}\right)\bigg] \\
& = \sum_{j=1}^{C}\bigg( \sum_{i=1}^{R}z_{ij} x^{(s)}_{ij}x^{(s')}_{ij}-\frac{\sum_{i=1}^{R} z_{ij}x^{(s)}_{ij}\sum_{i=1}^{R}z_{ij} x^{(s')}_{ij}}{n}\bigg)\end{align*}

\subsection{Simplication of covariance of \texorpdfstring{$\hat{\bbeta}$}{beta} }
\label{sec.covariance_of_beta_hat}
\begin{align*}
 & \cov{\hat{\bbeta}} \\
& =(X^{\tran} \tilde{X})^{-1} \tilde{X}^{\tran} \left( X \hat\Sigma_a X^\tran \bullet Z_a Z_a^\tran + X \hat\Sigma_b X^\tran \bullet Z_b Z_b^\tran  + \hat \sigma^2_e I_N \right)  \tilde{X}  (\tilde X^{\tran} X)^{-1} \\
& = 
(X^{\tran} \tilde{X})^{-1}  \left[ \tilde{X}^{\tran} \left( X \hat\Sigma_a X^\tran \bullet Z_a Z_a^\tran \right) \tilde X +  \tilde{X}^{\tran} \left( X \hat\Sigma_b X^\tran \bullet Z_b Z_b^\tran \right) \tilde{X} \right] (X^{\tran} X)^{-1} \\
&  \hspace{7cm} + \hat \sigma^2_e     (\tilde X^{\tran} X)^{-1} \\
\end{align*}

Let $P_a \in \mathbb{R}^{N \times N}$ denotes the permutation matrix that rearranges the data points so that they are grouped by clients and $P_b \in \mathbb{R}^{N \times N}$ denotes the permutation matrix that rearranges the data points so that they are grouped by items.

\begin{align}
\label{eq.sd_simp}
\tilde{X}^{\tran} \left( X \hat\Sigma_a X^\tran \bullet Z_a Z_a^\tran \right) \tilde X & = \tilde{X}^{\tran} P_a^\tran P_a \left( X \hat\Sigma_a X^\tran \bullet Z_a Z_a^\tran \right)  P_a^\tran P_a\tilde X \\
& = \tilde{X}^{\tran} P_a ^\tran \left(P_a X \hat\Sigma_a X^\tran P_a ^\tran  \bullet  P_a Z_a Z_a^\tran P_a^\tran \right) P_a \tilde X
\end{align}

In order to keep the notation simplified, we shall move ahead assuming that data is grouped by clients, i.e., we shall be denoting $P_a X$ by $X$, $P_a \tilde X$ by $\tilde X$ and $P_a Z_a$ by $Z_a$. The assumption simplifies the structure of $Z_a Z_a ^\tran$ to the block diagonal matrix given by

\[Z_a Z_a ^\tran = \begin{bmatrix}
1_{n_1} 1^\tran_{n_1} & 0 & \cdots & 0\\
0 & 1_{n_2} 1^\tran_{n_2} & \cdots & 0\\
0 & 0 & \ddots & 0\\
0 & 0 & 0 & 1_{n_R} 1^\tran_{n_R}\\
\end{bmatrix}\]

Now, let's define $U = X \hat{\Sigma_a}^{1/2}$ and represent column $i$ of $\tilde{X}$ by $s$ and column $j$ of $\tilde{X}$ by $t$. Therefore, $X \hat\Sigma_a X^\tran = \sum_{q = 0}^{p} U^{(q)} {U^{(q)}} ^\tran $ where $\{U^{(q)}\}$ represents columns of $U$.

 Therefore, $(i,j)$ entry in equation $\ref{eq.sd_simp}$ is given by 

\begin{align*} \sum_{q = 0} ^{p} (s^{\tran}_{1.},  \cdots, s^{\tran}_{R.}) & \begin{bmatrix}
{U^{(q)}}_{1.} {U^{(q)}}_{1.}^\tran & \cdots & 0\\
0  & \ddots & 0\\
0 & \cdots & {U^{(q)}}_{R.} {U^{(q)}}_{R.}^\tran\\
\end{bmatrix} \begin{pmatrix}
t_1\\
 \vdots\\
t_R\\
\end{pmatrix} \\
& = \sum_{q = 0} ^ p \sum_{i = 1}^{R} s_{i.} ^\tran {U^{(q)}}_{i.}  {U^{(q)}}_{i.}^\tran t_{i.}\\
& = \sum_{q = 0} ^ p \sum_{i = 1}^{R} s_{i.} ^\tran {U^{(q)}}_{i.}  {U^{(q)}}_{i.}^\tran t_{i.}\\
& = \sum_{q = 0} ^ p \sum_{i = 1}^{R} \sum_{j = 1}^C z_{ij} s_{ij} ^\tran {U^{(q)}}_{ij}  {U^{(q)}}_{ij}^\tran t_{ij}
 \end{align*}

Therefore, the quantity $\left[ \tilde{X}^{\tran} \left( X \hat\Sigma_a X^\tran \bullet Z_a Z_a^\tran \right) \tilde X +  \tilde{X}^{\tran} \left( X \hat\Sigma_b X^\tran \bullet Z_b Z_b^\tran \right) \tilde{X} \right] \in \mathbb{R}^{(p+1) \times (p+1)}$ can be computed in $\mathcal{O}(N p^3)$ time which makes the computation of $\cov(\hat{\bbeta})$ possible in time linear in $N$. 
\subsection{Proof of Lemma~\ref{lemma:var_E}}\label{sec:proof-lemma-refl}
\begin{align*}
& q_a^k  (a_1, \cdots, a_R) \propto \Prob (y, \{\bsa_i\}_{i=1}^R | \bsb_j = E_{Q^{(k-1)}} [\bsb_j])\\
& \propto \exp \left \{ -\frac{1}{2} \sum_{i=1}^R \sum_{j=1}^C \frac{z_{ij} (y_{ij} -\mathbf{x}_{ij} ^\tran (\bsa_i + \mu_{b,j}^{(k-1)}+ \bbeta))^2}{\left(\sigma^{(k-1)}_e\right)^2} -\frac{1}{2} \sum_{i=1}^R \bsa_i ^\tran (\Sigma_a^{(k-1)})^{-1} \bsa_i \right \}\\
& \propto \exp \Bigg \{ -\frac{1}{2} \sum_{i=1}^R \bsa_i^\tran \Bigg( (\Sigma^{(k-1)}_a)^{-1} +\frac{\sum_{j=1}^{C}z_{ij}\mathbf{x}_{ij}  \mathbf{x}_{ij} ^{\tran}}{\left(\sigma^{(k-1)}_e\right)^2} \Bigg) \bsa_i\\
& \hspace{2cm} + \bsa_i ^\tran \frac{\sum_{j=1}^{C}z_{ij}\big(y_{ij} - \mathbf{x}_{ij} ^{\tran}\big(\bbeta^{(k-1)} +\mu^{(k-1)}_{b,j}\big)\big)\mathbf{x}_{ij} }{\left(\sigma^{(k-1)}_e\right)^{2}} \Bigg\} \\
\end{align*}
If $q_a^k (a_1, \cdots, a_R)$ represents the distribution $\mathcal{N}(\mu^{(k)}_{a,i}, \Sigma^{(k)}_{a,i})$, then
\begin{equation}
q_a^k (a_1, \cdots, a_R) \propto \exp \left \{ -\frac{1}{2} \bsa_i ^\tran {\left(\Sigma^{(k)}_{a,i}\right)}^{-1} \bsa_i + \bsa_i ^\tran {\left(\Sigma^{(k)}_{a,i}\right)}^{-1} \mu^{(k)}_{a,i} \right \}
\end{equation}
Thus,
\begin{equation*}
\mu^{(k)}_{a,i} = \Sigma^{(k)}_{a,i} \left(\frac{\sum_{j=1}^{C}z_{ij}\big(y_{ij} - \mathbf{x}_{ij} ^{\tran}\big(\bbeta^{(k-1)} +\mu^{(k-1)}_{b,j}\big)\big) \mathbf{x}_{ij}}{\left(\sigma^{(k-1)}_e\right)^{2}}\right)
\end{equation*}
and
\begin{equation*}
\Sigma^{(k)}_{a,i} = \left((\Sigma^{(k-1)}_a)^{-1} +\frac{\sum_{j=1}^{C}z_{ij}\mathbf{x}_{ij} \mathbf{x}_{ij}^{\tran}}{\left(\sigma^{(k-1)}_e\right)^2}\right)^{-1}.
\end{equation*}
Similarly,
\begin{align*}
q_b^k & (b_1, \cdots, b_C) \propto \Prob (y, \{\bsb_j\}_{j=1}^C | \bsa_i = E_{Q^{(k)}} [\bsa_i])\\
& \propto \exp \left \{ -\frac{1}{2} \sum_{j=1}^C \sum_{i=1}^R \frac{z_{ij} (y_{ij} - \mathbf{x}_{ij}^\tran (\bsb_j + \mu_{a,i}^{(k)}+ \bbeta))^2}{\left(\sigma^{(k-1)}_e\right)^2} -\frac{1}{2} \sum_{j=1}^C \bsb_j ^\tran (\Sigma_b^{(k-1)})^{-1} \bsb_j \right \}\\
& \propto \exp \Bigg \{ -\frac{1}{2} \sum_{j=1}^C \bsb_j^\tran \Bigg( (\Sigma^{(k-1)}_b)^{-1} +\frac{\sum_{i=1}^{R}z_{ij}\mathbf{x}_{ij} \mathbf{x}_{ij}^{\tran}}{\left(\sigma^{(k-1)}_e\right)^2} \Bigg) \bsb_j\\
& \hspace{2cm} + \bsb_j ^\tran \frac{\sum_{j=1}^{C}z_{ij}\big(y_{ij} - \mathbf{x}_{ij}^{\tran}\big(\bbeta^{(k-1)} +\mu^{(k)}_{a,i}\big)\big)\textbf{x}_{ij}}{\left(\sigma^{(k-1)}_e\right)^{2}} \Bigg\} \\
\end{align*}
For the same reasoning,
\begin{equation*}
\mu^{(k)}_{b,j} = \Sigma^{(k)}_{b,j} \left(\frac{\sum_{i=1}^{R}z_{ij}\big(y_{ij} - \textbf{x}_{ij}^{\tran}\big(\bbeta^{(k-1)} +\mu^{(k-1)}_{a,i}\big)\big)\textbf{x}_{ij}}{\left(\sigma^{(k-1)}_e\right)^{2}}\right)
\end{equation*}
and
\begin{equation*}
\Sigma^{(k)}_{b,j} = \left((\Sigma^{(k-1)}_b)^{-1} +\frac{\sum_{i=1}^{R}z_{ij}\textbf{x}_{ij} \textbf{x}_{ij}^{\tran}}{\left(\sigma^{(k-1)}_e\right)^2}\right)^{-1}
\end{equation*}

%% file: plots_appendix_randomslopes.tex
\label{sec.generalized_results}
For the general setup, we consider the covariates to be the union of item-based features and client-based features, i.e., \( x_{ij} = (1, v_j, u_i) \). We include random effects for clients on item-based features and vice versa. Therefore, the model is given by
\[
y_{ij} = x_{ij}^\top \beta + v_j^\top a_i + u_i^\top b_j + \varepsilon_{ij},
\]
where \( a_i \overset{\text{i.i.d.}}{\sim} \mathcal{N}_3(0, \Sigma_a) \), \( b_j \overset{\text{i.i.d.}}{\sim} \mathcal{N}_3(0, \Sigma_b) \), and \( \varepsilon_{ij} \overset{\text{i.i.d.}}{\sim} \mathcal{N}(0, \sigma^2_e) \), with
\[
\Sigma_a = \begin{bmatrix}
1 & 0.2 & 0.2\\
0.2 & 1 & 0.2\\
0.2 & 0.2 & 1
\end{bmatrix}, \quad
\Sigma_b = \begin{bmatrix}
1 & 0.2 & 0.2\\
0.2 & 1 & 0.2\\
0.2 & 0.2 & 1
\end{bmatrix}, \quad \text{and} \quad \sigma^2_e = 1.
\]

We chose \( \beta = (0.1, 0.2, 0.3, 0.4, 0.5) \), and generated \( (v_j, u_i) \overset{\text{i.i.d.}}{\sim} \mathcal{N}_4(0, I) \).

In Figures~\ref{fig:beta_and_sigma2e_general_setup} and~\ref{fig:sigma_ab_general_setup}, we present the consistency of our algorithm in estimating the model parameters under the general setup. The results are similar to those presented earlier, i.e., the accuracy of our proposed algorithm approaches that of the maximum likelihood estimates as \( N \) increases. In Figure~\ref{fig:estimation_time_general_setup}, we compare the estimation time of our proposed algorithm with existing approach in the general setup; once again, the conclusions are consistent with those reported previously.

\begin{figure}[ht]
\centering
\includegraphics[width=\linewidth]{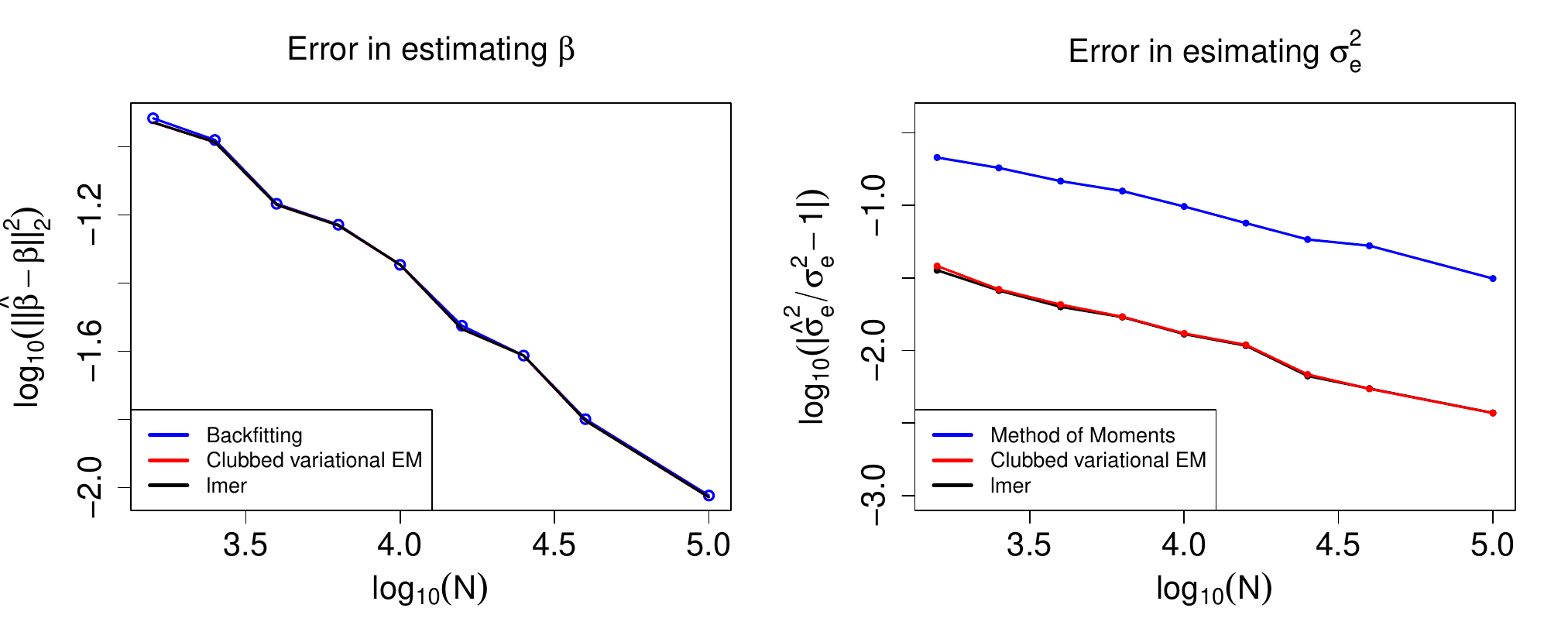}
\caption{General setup: Comparison of accuracy in estimating $\bbeta$ and $\sigma^2_e$ by variational EM (proposed algorithm) and lmer(existing algorithm). The errors in estimating $\bbeta$ are close to each other for the illustrated algorithms, which leads the lines in the plot to coincide.}
\label{fig:beta_and_sigma2e_general_setup}
\end{figure}
\begin{figure}[ht]
\centering
\includegraphics[width=\linewidth]{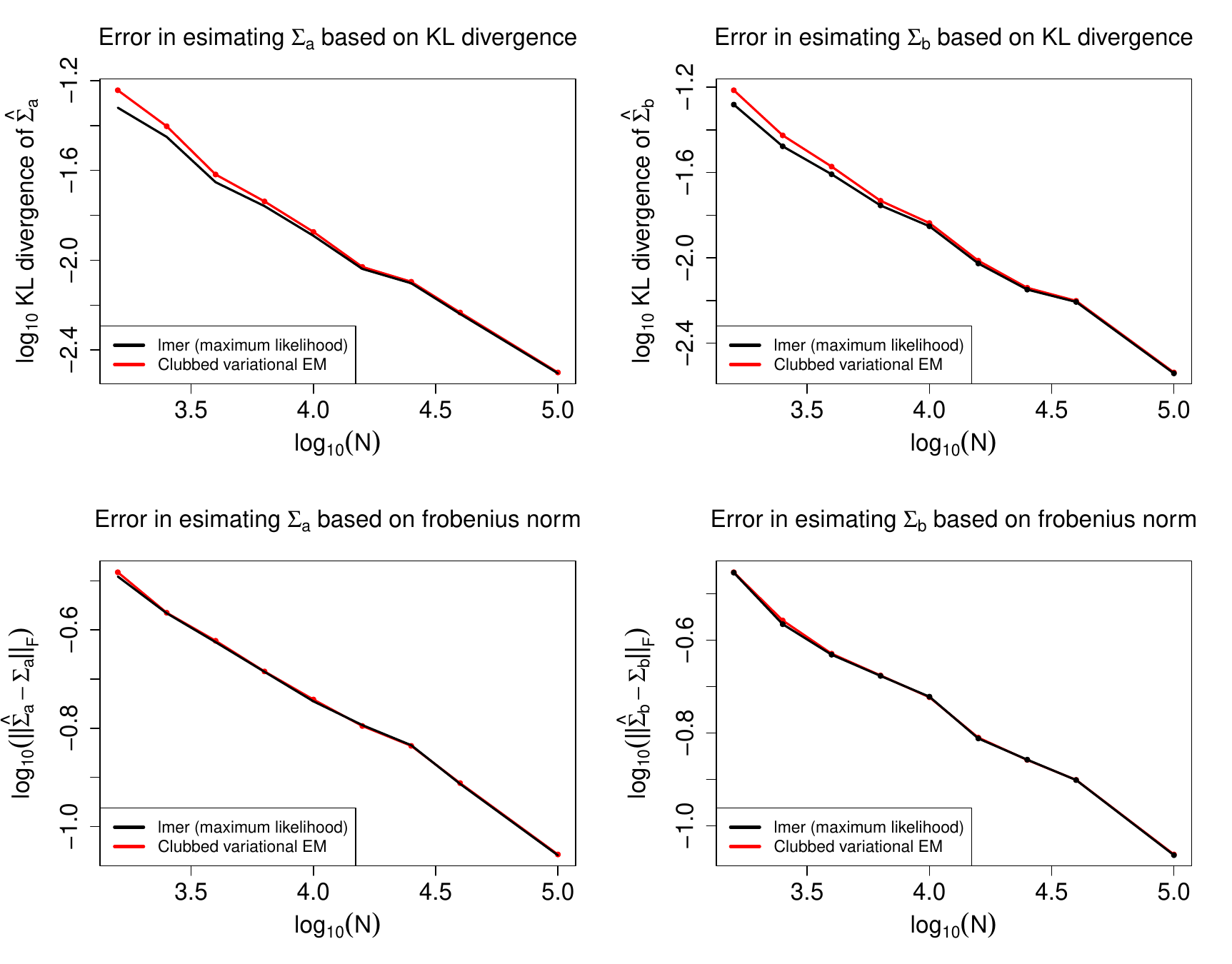}
\caption{General setup: Comparison of accuracy in estimating ${\Sigma}_a$ and ${\Sigma}_b$ based on KL divergence and Frobenius norm by variational EM (proposed algorithm) and lmer(existing algorithm).}
\label{fig:sigma_ab_general_setup}
\end{figure}
\begin{figure}[ht]
\centering
\includegraphics[scale=0.5]{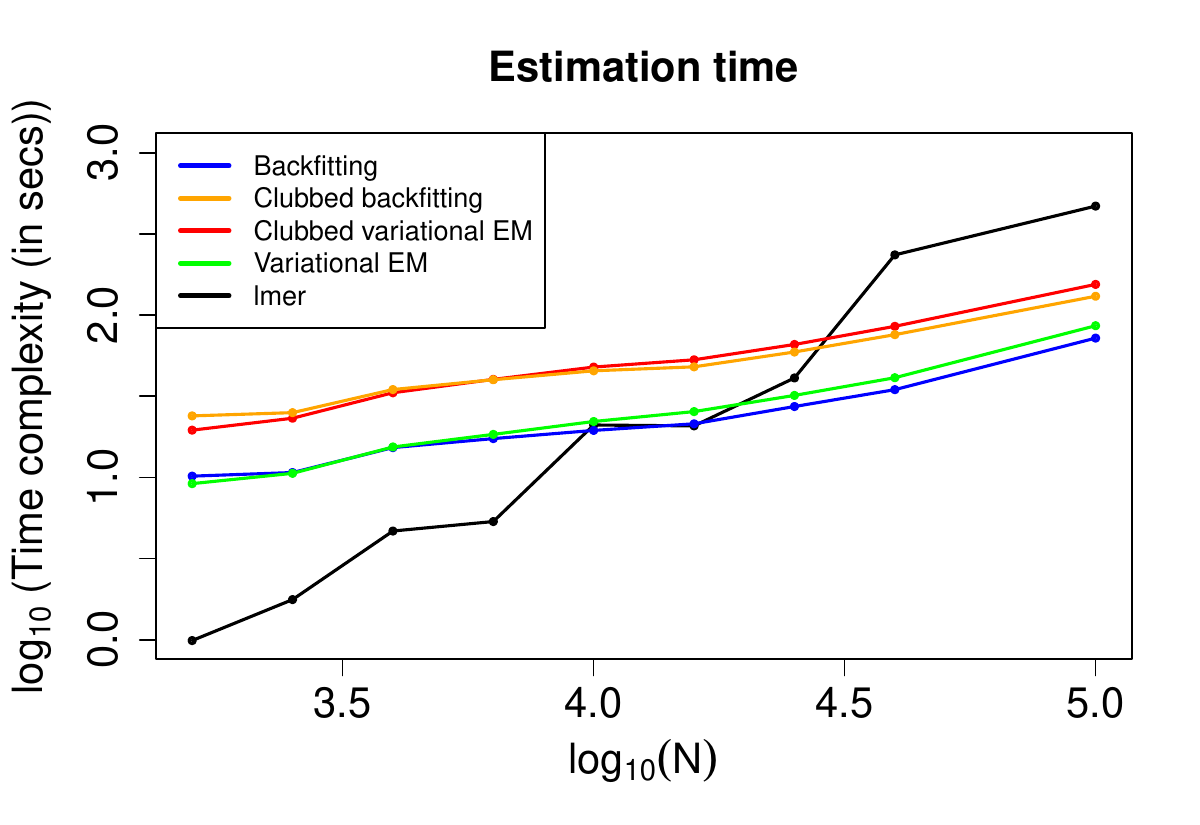}
\caption{General setup: Comparison of estimation time by backfitting, clubbed backfitting, variational EM, and clubbed variational EM (four proposed algorithms) and ``lmer"(existing algorithm).}
\label{fig:estimation_time_general_setup}
\end{figure}